\documentclass[11pt]{article}

\usepackage{amsfonts,amssymb,amsmath,latexsym,ae,aecompl,bbm,algorithm,algpseudocode}

\usepackage{tikz}

\newcommand{\F}{\vspace*{\smallskipamount}}
\newcommand{\FF}{\vspace*{\medskipamount}}
\newcommand{\FFF}{\vspace*{\bigskipamount}}
\newcommand{\B}{\vspace*{-\smallskipamount}}
\newcommand{\BB}{\vspace*{-\medskipamount}}
\newcommand{\BBB}{\vspace*{-\bigskipamount}}

\newcommand{\cA}{{\mathcal A}}

\newcommand{\cE}{{\mathcal E}}

\newcommand{\cO}{{\mathcal O}}

\newcommand*\circled[1]{\tikz[baseline=(char.base)]{
            \node[shape=circle,draw,inner sep=1pt] (char) {$#1$};}}
            
\newcommand{\Item}{\B\item}
\newcommand{\Paragraph}[1]{\BBB\paragraph{#1}}
\newcommand{\remove}[1]{}

\newlength{\pagewidth}
\newlength{\captionwidth}
\newcommand{\qed}{\hfill $\square$ \smallbreak}
\newenvironment{proof}{\noindent{\bf Proof:}}{\qed}

\newcommand{\RB}{\raisebox{2.5ex}{~}}
\newcommand{\LB}{\raisebox{-1.5ex}{~}}

\newtheorem{theorem}{Theorem}
\newtheorem{lemma}{Lemma}

\newtheorem{proposition}{Proposition}

\begin{document}

\baselineskip 	3ex
\parskip 		1ex

\title{			Broadcasting in Ad Hoc Multiple Access Channels~\footnotemark[1]\FFF\FFF}

\author{		Lakshmi Anantharamu \footnotemark[2]	
			\and 
			Bogdan S. Chlebus \footnotemark[2]}


\date{}

\maketitle

\footnotetext[1]{This paper was published in a preliminary form as~\cite{AnantharamuC-SIROCCO13} and its final version appeared as~\cite{AnantharamuC15}.}

\footnotetext[2]{Department of Computer Science and Engineering, University of Colorado Denver, Denver, Colorado 80217, USA.
This work was supported by the NSF Grant 1016847.}

\FFF


\begin{abstract}
We study broadcast in multiple access channels in dynamic  adversarial settings.
There is an unbounded supply of anonymous stations attached to a synchronous channel.
There is an adversary who  injects packets into stations to be broadcast on the channel.
The adversary is restricted by injection rate, burstiness, and by how many passive stations can be simultaneously activated by providing them with packets.
We consider deterministic distributed broadcast algorithms, which are further categorized by their properties.
We investigate for which injection rates can algorithms attain bounded packet latency,  when adversaries are restricted to be able to activate at most one station per round.
The rates of algorithms we present make the increasing sequence consisting of $\frac{1}{3}$, $\frac{3}{8}$ and~$\frac{1}{2}$, reflecting the additional features of algorithms.
We show that injection rate~$\frac{3}{4}$ cannot be handled with bounded packet latency.

\FFF\FFF

\end{abstract}

\vfill

\thispagestyle{empty}
\setcounter{page}{0}

\newpage

\section{Introduction}

\label{sec:introduction}

Multiple access channels model shared-medium networks in which a simultaneous broadcast to all users is provided by architecture. 
They are an abstraction of the networking technology of the popular implementation of local area networks by the Ethernet suite of technologies~\cite{MetcalfeB-CACM76}.

In a multiple access channel, transmissions by multiple users that overlap in time result in interference so that none can be  received successfully.
This makes it necessary either to avoid  conflict for access to the channel altogether or to have a mechanism to resolve conflict when it occurs.

We consider broadcasting in multiple-access channels in a dynamic scenario when there are many stations but only a few of them are active at any time and the stations' status of active versus passive may change. 
This corresponds to a realistic situation when most stations are idle for most of the time, while a few stations occasionally want to use the broadcast functionality of the channel.
Moreover, it is normally impossible to predict in advance which stations will need to access the channel at what times, as bursts of activity among stations do not exhibit any regular patterns.
Resolving conflicts for access to a channel can be accomplished by using randomization, as is implemented in the carrier-sense multiple access, see \cite{Keshav-book97}.

Considering deterministic algorithms and their worst-case performance requires a  methodological setting specifying worst-case bounds on how much traffic a network would need to handle.
This can be accomplished formally through suitable adversarial models of  demands on network traffic.
Adversarial models normally assume synchronous channels with stations acting in lockstep, which is the case of this work.

Another component in a specification of a broadcast system is how much knowledge about the  system can communicating agents use in their codes of algorithms.
Various approaches are possible to model multiple access channels in terms of what is known to the stations attached to the communication medium.
Historically, the first approach was to use the queue-free model, in which each injected packet is treated as if handled by an independent station without any name and no private memory for a queue.
In such a model, the number of stations is not set in any way, as stations come and go similarly as packets do; see \cite{Gallager-TIT85} for the initial work on this model, and~\cite{BenderFHKL05} for more recent one.

An alternative approach is to have a system with a fixed number~$n$ of stations, each equipped with a private memory to store packets in a queue.
An attractive feature of such fixed-size systems is that even simple randomized protocols like Aloha are stable under suitable traffic load~\cite{TsybakovM79}, while in the queue-free model the binary exponential backoff is unstable for any arrival rate~\cite{Aldous-TIT87}.

Popular assumptions used in the literature addressing distributed deterministic broadcasting stipulate that there are some $n$ stations attached to a channel and that each station is identified by its name in the interval $[0,n-1]$, with each station knowing the number~$n$ and its own name; see~\cite{AnantharamuCKR-INFOCOM10, AnantharamuCKR-SIROCCO11, AnantharamuCR-OPODIS09, ChlebusKR09, ChlebusKR-TALG12}.
Our goal is to explore deterministic broadcasting on multiple-access channels when there are many stations attached to a channel but only a few stations use it at a time.
In such circumstances, using names permanently assigned to stations by deterministic distributed algorithms may create an unnecessarily large overhead measured as packet latency and queue size.
This is because the  packet latency is expected to  be a function of the total number of stations in the system, which could be arbitrarily large; see~\cite{AnantharamuCKR-INFOCOM10, AnantharamuCKR-SIROCCO11} for such approaches.

In this paper, we consider distributed deterministic broadcasting which departs from the assumption about a fixed known size of the system.
Instead, we view the system as consisting of a very large set of stations which are not individually identified in any way.
Stations  join the activity on the channel when they want to broadcast, which needs to be coordinated with stations that are already active.
Such coordination could be associated with the medium-access control layer \cite{Keshav-book97} of a layered protocol design.

The process of activating stations is considered without assuming that there is a finite fixed set of stations attached to the channel, that their number is known to each participating station, and that each station has a unique name which it knows. 
We call such channels \emph{ad hoc} to emphasize the volatility of the system and the relative lack of knowledge of individual stations about themselves and the environment.
Ad hoc channels are a crossover between the queue-free model, with which they share the property of an unbounded supply of anonymous stations activated by injected packets, and the model of finitely many stations in a system, with which they share the property that stations use their private memories to implement private queues to store pending packets.

We measure the performance of broadcast algorithms by packet latency and queue sizes. 
These metrics reflect the constraints on packet injection imposed by the adversarial model.
Such constraints include packet injection rate, understood as the average number of packets injected in a large time interval, and burstiness, which means the maximum number of packets that can be injected simultaneously.
Adversarial models of traffic allow to study the worst-case performance of deterministic communication algorithms.


\begin{table}
\begin{center}
\begin{tabular}{|c ||c |c |c |c |c |}
\hline
\RB \LB
& $\mathsf{\rho < \frac{1}{3}}$ & $\mathsf{\rho \le \frac{3}{8}}$ & $\mathsf{\rho \le \frac{1}{2}}$ 
&$\mathsf{\frac{1}{2}<\rho < \frac{3}{4}}$&$\mathsf{\rho \ge \frac{3}{4}}$\\
\hline\hline
\RB \LB
\textsf{non-adaptive activation-based } & $\frac{3b-3}{1-3\rho}$ & ? &  && \\
\hline
\RB
\textsf{non-adaptive full-sensing }&& $2b +4$ & ? && \\
\hline
\RB \LB
\textsf{adaptive activation-based } & & &$4b-4$&& \\
\hline
\RB \LB
\textsf{adaptive full-sensing } & & &&?& impossible\\
\hline
\end{tabular}
\parbox{\pagewidth}{
~
\caption{\label{tab:collision-detection}
Some results for channels \emph{with} collision detection.
Entries that are expressions are upper bounds on packet latency in terms of type $(\rho,b)$ of a $1$-activating adversary.
Question marks indicate open questions  if deterministic algorithms with bounded packet latencies exist, satisfying the respective restrictions.
The impossibility is of existence of a deterministic algorithm with bounded packet latency.
}}
\end{center}
\end{table}

\Paragraph{Our results.}

We propose an adversarial model of traffic demands for ad hoc synchronous multiple access channels, which represents \emph{dynamic} environments in which stations freely join and leave broadcasting activity.
To make an anonymous system able to break symmetry in a deterministic manner, we restrict adversaries by allowing them only to activate at most one station per round.
This is shown sufficient to demonstrate existence of \emph{deterministic} distributed broadcast algorithms with a bounded packet latency, subject to restrictions on injection rates.
In this overview of results we refer to technical terms that are precisely defined in Section~\ref{sec:preliminaries}.

We categorize algorithms into acknowledgment based, activation based and finally general algorithms, which are called full sensing. 
Independently from that, we differentiate algorithms by the property if they use control bits in messages or not, calling them adaptive and non-adaptive,  respectively.
We give a number of algorithms, for channels with and without collision detection, for which we assess injection rates they can handle with bounded packet latency. 
More specifically, our non-adaptive  activation-based algorithm can handle injections rates smaller  than~$\frac{1}{3}$ on channels with collision detection, a  non-adaptive  full-sensing algorithm can handle injection rate~$\frac{3}{8}$ on  channels with collision detection, and an adaptive  activation-based algorithm can handle  injection rate~$\frac{1}{2}$ on channels without collision detection.
These positive results are summarized in Tables~\ref{tab:collision-detection} and~\ref{tab:no-collision-detection}.

We show a number of impossibility results as well.
First, no deterministic distributed algorithm can guarantee that each injected packet is eventually heard on the channel, when executed against an adversary that can activate multiple stations at the same time and when the adversary's burstiness is at least~$2$.
Second, no acknowledgment-based algorithm is fair against a $1$-activating adversary for any injection rate $0<\rho<1$, for a sufficiently large burstiness that depends on this~$\rho$.
Third, no deterministic algorithm for channels with collision detection  can provide bounded packet latency when injection rates are at least~$\frac{3}{4}$.

Table~\ref{tab:no-collision-detection} has only two classes of algorithms, namely two kinds of adaptive algorithms, while Table~\ref{tab:collision-detection} has four categories of algorithms.
This is because it is as an open question if a \emph{non-adaptive} deterministic algorithm for channels \emph{without} collision detection can provide bounded packet latency when executed against adversaries constrained by \emph{any} positive upper bound on their injection rates, where such a bound could  be  as small as possibly needed.


\begin{table}
\begin{center}
\begin{tabular}{|c ||c |c |c |}
\hline
\RB \LB
& $\mathsf{\rho \le \frac{1}{2}}$ 
&$\mathsf{\frac{1}{2}<\rho < \frac{3}{4}}$&$\mathsf{\rho \ge \frac{3}{4}}$\\
\hline\hline
\RB \LB
\textsf{adaptive activation-based }  &$4b-4$&& \\
\hline
\RB \LB
\textsf{adaptive full-sensing }  &&?& impossible\\
\hline
\end{tabular}
\parbox{\pagewidth}{
~
\caption{\label{tab:no-collision-detection}
Some results for adaptive algorithms on channels \emph{without} collision detection.
The entry $4b-4$ is an upper bound on packet latency in terms of type $(\rho,b)$ of a $1$-activating adversary.
The question mark indicates an open question if an algorithm with a bounded packet latency exists,  satisfying the respective restrictions.
The impossibility is of existence of a deterministic algorithm with bounded packet latency.
}}
\end{center}
\end{table}

\Paragraph{Related work.}

The adversarial queuing methodology in communication algorithms was introduced by Borodin et al.~\cite{BorodinKRSW-JACM01} and Andrews et al.~\cite{AndrewsAFLLK-JACM01}, who  used it to study the stability of store-and-forward routing in wired networks.
Adversarial queueing on multiple access channels was first studied by Bender et al.~\cite{BenderFHKL05}, who considered  randomized algorithms for the queue-free model.
A deterministic distributed broadcasting on multiple access channels with queues in adversarial settings was investigated by Chlebus et al.~\cite{ChlebusKR09,ChlebusKR-TALG12} and by Anantharamu et al.~\cite{AnantharamuCKR-INFOCOM10, AnantharamuCKR-SIROCCO11, AnantharamuCR-OPODIS09}.
That work on deterministic distributed algorithms was about systems with a known number of stations attached to the channel and with stations using individual names.

Acknowledgment-based algorithms include the first randomized algorithms studied on dynamic channels, as Aloha and binary exponential backoff fall into this category.
The throughput of multiple access channels, understood as the maximum injection rate with Poisson traffic that can be handled by a randomized algorithm and make the system stable (ergodic), has been intensively studied in the literature.
It was shown to be at most $0.568$ by Tsybakov and Likhanov~\cite{TsybakovL87}. 
Goldberg et al.~\cite{GoldbergJKP04} gave related bounds for backoff, acknowledgment-based and full-sensing algorithms.
H{\aa }stad et al.~\cite{HastadLR-SICOMP96} compared polynomial and exponential backoff algorithms in the queuing model with respect to bounds on their throughput.
For an  account of early work on full-sensing algorithms in channels with collision detection in the queue-free model, see the survey by Gallager~\cite{Gallager-TIT85}.

Randomized algorithms of bounded packet latency were given by Raghavan and Upfal~\cite{RaghavanU-SICOMP98} in the queuing model and by Goldberg et al.~\cite{GoldbergMPS-JACM00} in the queue-free model.
Upper bounds on packet latency in adversarial networks was studied by Anantharamu et al.~\cite{AnantharamuCKR-INFOCOM10,AnantharamuCKR-SIROCCO11} in the case of multiple access channels with injection rate less than~$1$ and by Ros{\'e}n and Tsirkin~\cite{RosenT06} for general networks and adversaries of rate~$1$.

Deterministic algorithms for collision resolution in static algorithmic problems on multiple access channels were first considered by Greenberg and Winograd~\cite{GreenbergW-JACM85} and Koml{\'o}s and Greenberg~\cite{KomlosG-TIT85}.
Algorithmic problems of distributed-computing flavor in systems in which multiple access channels provide the underlying communication infrastructure were considered by Bie\'nkowski et al.~\cite{BienkowskiKKK-STACS10} and Czy\.zowicz et al.~\cite{CzyzowiczGKP11}.

\section{Technical preliminaries}

\label{sec:preliminaries}

A multiple-access channel consists of a shared communication medium and stations attached to it.
We consider dynamic broadcasting, in which packets are injected into stations continually and the goal is to have them successfully transmitted on the channel.

A \emph{message} transmitted by a station includes at most one packet and some control bits, if any.
Every station receives a transmitted message successfully, including the transmitting station, when the transmission of this message does not overlap with transmissions by other stations; in such a case we say that the message is \emph{heard} on the channel.

We consider slotted channels which operate in globally synchronized rounds.
Messages and rounds are calibrated such that transmitting one message takes one full round.
A message transmitted in a round is delivered to every station in the same round.
When at least two messages are transmitted in the same round then this creates a \emph{collision}, which prevents any station from hearing any of the transmitted messages. 

When no station transmits in a round, then the round is called \emph{silent}.
A channel is said to be \emph{with collision detection} when the feedback from the channel in a collision round is different from the feedback received during a silent round,  otherwise the channel is \emph{without collision detection}. 
For a channel without collision detection, a collision round and a silent one are perceived  the same. 
A round is \emph{void} when no station hears a message; such a round is either silent or collision.

\Paragraph{Ad hoc channels.}

A station is said to be \emph{active}, at a point in time, when it has pending packets that have not been heard on the channel yet.
A station is  \emph{passive} in a round if either it has never had any packets to broadcast or all the packets it has ever received to broadcast have already been heard on the channel in previous rounds.

We assume that there is an unbounded supply of passive stations.
A passive station is said to get \emph{activated} when a packet or multiple packets are injected into it. 
We impose quantitative restrictions on how passive stations may be activated in a round, which results in finitely many stations being active in any round.
There is no upper bound on the number of active stations in a round of an infinite execution, since there is an unbounded supply of passive stations.

Stations are \emph{anonymous} when there are no individual names assigned to them.
We consider channels that are \emph{ad hoc} which means the following three properties: 
\begin{enumerate}
\item[(1)]
every station is anonymous, 
\item[(2)] 
an execution starts with every station initialized as passive, and 
\item[(3)] 
there is an unbounded supply of passive stations.
\end{enumerate}

\Paragraph{Adversarial model of packet injection.}

Packets are injected by leaky-bucket adversaries.
For a number $0<\rho\le 1$ and integer $b>0$, the \emph{adversary of  type $(\rho,b)$} may inject at most $\rho |\tau| + b$ packets in any time interval $\tau$ of $|\tau|$ rounds.
In such a context, the number~$\rho$ is called the \emph{rate of injection}.
The maximum number of packets that an adversary may inject in one round is called the \emph{burstiness} of this adversary.
The adversary of type $(\rho,b)$ has burstiness $\lfloor \rho + b\rfloor$.

Adversaries we consider are constrained by how many stations they can activate in a round.
An adversary is \emph{$k$-activating}, for an integer $k>0$, if at most $k$ stations may be activated in a round.
We consider $1$-activating adversaries, unless explicitly stated otherwise, which is necessary if algorithms are to be deterministic, see Section~\ref{sec:limitations}.

\Paragraph{Broadcast algorithms.}

We consider  broadcast algorithms that are deterministic and distributed.
In the context of distributed communication algorithms,  ``knowledge'' of properties of a system means using such properties as a part of code of an algorithm.
In algorithms we consider, names of stations and the number of stations in the system are not known.
This is in contrast with previous work on deterministic distributed algorithms, see~\cite{AnantharamuCKR-INFOCOM10, AnantharamuCKR-SIROCCO11, AnantharamuCR-OPODIS09, ChlebusKR09, ChlebusKR-TALG12}, where the names of stations and the number of stations could be used in a code.
No information about adversaries is reflected in codes of algorithms considered in this paper.

Every station has a private memory to store data needed in executing a communication algorithm.
This memory is considered to be unbounded, in the sense that it may store an arbitrary amount of data.
The part of a private memory of a station used to store packets pending transmission is organized as a queue operating  in a first-in-first-out manner.
Packets broadcast successfully are removed from their queues and discarded.
Packets are never dropped before a successful broadcast.

A \emph{state} of a station is determined by  values of its private variables, with the exception of  the queue to store packets, which is not a part of states.
One state is distinguished as \emph{initial}.
An execution begins with every station being in such initial state and with an empty queue.  
Packets are treated as objects devoid of properties, in that their contents do not affect state transitions.

The algorithms we consider are distributed, in the sense that they are ``event driven.''
An \emph{event}, in which a station participates, consists of everything that happens to the station in a round, including what the station receives as feedback from the channel and how many packets are injected into it.

An event is structured as the following sequence of actions occurring in a round in the order given:
\begin{enumerate}
\item[(i)] 
transmitting a packet,
\item[(ii)] 
receiving a feedback from the channel, 
\item[(iii)]  
having new packets injected,
\item[(iv)] 
making a state transition.
\end{enumerate}
Some among the actions (i) and (iii) may be void in a station in a round.
A state transition depends on a current state, a feedback from the channel, and on whether new packets were injected in the round.
In particular, the following actions occur during a state transition.
If a packet has just been successfully transmitted then it is dequeued and discarded.
If new packets have just been injected then they are all enqueued.
If a message is to be transmitted in the next round, possibly subject to packet availability,  then a message to be transmitted is prepared.
Such a message may include the packet from the top of the queue, when the queue is nonempty, but a message may consist of only control bits.

A station that begins a round as active with only one pending packet becomes passive when the packet is transmitted and heard in this round.
Since stations are anonymous and there is an unbounded supply of passive stations, we may assume without loss of generality that an active station is not activated again after becoming passive.

\Paragraph{Classes of algorithms.}

We define subclasses of algorithms by specifying what can be included in messages and how state transitions occur. 
We begin with the categorization into full-sensing, activation-based and acknowledgment-based algorithms.

General algorithms are called \emph{full sensing}.
This means that state transitions occur in each round, according to the state-transition rules represented by code.
The term ``full sensing'' is to indicate that every station is sensing the channel in every round. 
This encompasses passive stations, which means that when a full-sensing algorithm is executed, then passive stations undergo state transitions from the beginning of execution.

Algorithms such that a passive station is always in an initial state are called \emph{activation based}.
These algorithms have stations ignore feedback from the channel when they do not have any  packets to broadcast.

Finally, algorithms such that a passive station is in an initial state and an active station resets its state to initial in a round in which a packet it transmits is heard are called \emph{acknowledgment based}.
This definition is correct due to the stipulation that the contents of queues do not belong to what constitutes a state; in particular, a station may be in an initial state when its queue is nonempty.
This class of algorithms is inspired by randomized algorithms like Aloha or backoff ones.
Randomized algorithms in this class are prominent in applications, but deterministic acknowledgment based algorithms in ad hoc channels are too restricted to be useful, as we demonstrate in Section~\ref{sec:limitations}.

A station executing a full-sensing algorithm may, at least in principle, remember the whole history of feedback from the channel, unless the size of its private memory restricts it in this respect, which is not the case in our considerations.
An active station executing an activation-based algorithm may remember the history of  feedback from the channel since activation.
An active station executing an acknowledgment-based algorithm may remember the history of feedback from the channel since activation or the latest successful transmission if such occurred.
We understand these categorizations so that an acknowledgment-based algorithm is activation based, and an activation-based algorithm is full sensing.
This is because a station executing an activation-based algorithm could be considered as receiving feedback from the channel but idling in an initial state while not having pending packets.

When control bits are used in messages then we say that an algorithm is \emph{adaptive}, otherwise the algorithm is \emph{non-adaptive}.
The categorization of adaptive versus non-adaptive is independent of the categorization into full sensing and activation based and acknowledgment based, so we have six categories of algorithms in total.
This categorization of algorithms holds independently of categorization of  channels into ones with and without collision detection.
The strongest algorithms that we consider are full sensing adaptive for channels with collision detection, while the weakest algorithms are acknowledgment-based non-adaptive for channels without collision detection.

The terminology about acknowledgment-based and full-sensing algorithms is consistent with that used in the literature on randomized protocols in the queue-free model, see~\cite{Gallager-TIT85}, and also with the terminology used in the recent literature on deterministic distributed algorithms in adversarial settings, see~\cite{AnantharamuCKR-INFOCOM10,AnantharamuCKR-SIROCCO11, AnantharamuCR-OPODIS09, ChlebusKR09,ChlebusKR-TALG12}.
The categorization of algorithms as activation based is new because it is applicable for ad hoc channels.

\Paragraph{The quality of broadcasting.}

An execution of an algorithm is said to be \emph{fair} when each packet injected into a station is eventually heard.
An algorithm is fair against an adversary when each of its executions is fair when packets are injected subject to the constrains defining the adversary.

An execution of an algorithm has \emph{at most $Q$ packets queued} when in each round the number of packets stored in the queues of the active stations is at most~$Q$.
We say that an algorithm has at most $Q$ packets queued, against the adversary of a given type, when at most $Q$ packets are queued in any execution of the algorithm against such an adversary.

An algorithm is \emph{stable}, against the adversary of a given type, when there exist an integer $Q$ such that at most~$Q$ packets are queued in any execution against this adversary.
When an algorithm is unstable then the queues may grow unbounded in some executions, but no packet is ever dropped unless heard on the channel.
The semantics of multiple access channels allows for at most one packet to be heard in a round.
This means that when the injection rate of an adversary is greater than~$1$ then for any algorithm some of its executions produce unbounded queues.
In this paper, we consider only injection rates that are at most~$1$.

An execution of an algorithm has \emph{packet latency~$t$} when each packet spends at most $t$ rounds in the queue before it is heard on the channel.
We say that an algorithm has packet latency~$t$ against the adversary of a given type when each execution of the algorithm against this adversary  has packet latency~$t$.

\section{Limitations on deterministic broadcasting}

\label{sec:limitations}

In this section, we consider what limitations on deterministic distributed broadcasting are inherent in the properties of ad-hoc multiple access channels and the considered classes of algorithms.

We begin with restrictions on the number of stations that can be activated in one round when broadcast algorithms are deterministic.
We show that this needs to be at most one station.
The idea of proof is to construct an execution in which there are two stations that proceed through the same states.


\begin{proposition}
\label{proposition:2-activating}

No deterministic distributed algorithm is fair on channels with collision detection against a $2$-activating adversary of burstiness at least~$2$.
\end{proposition}

\begin{proof} 
Let us consider an arbitrary deterministic distributed algorithm.
We will specify an execution in which two stations perform the same actions.
The execution is determined by the given algorithm and by how the adversary injects packets, which we can specify.

Let an execution begin with the adversary injecting packets simultaneously into two passive stations, one packet per station.
These two anonymous stations execute the same deterministic algorithm, so their actions are the same until one station experiences what the other does not.
The adversary does not need to inject any other packets.
It follows, by induction on the round numbers, that the stations undergo the same state transitions.
In particular, when one of these two stations transmits a packet then the other one transmits as well, and when one station pauses then the other station pauses as well.
In this execution, each transmission attempt results in a collision.
This means that the two packets never get heard on the channel.
\end{proof}

In the light of Proposition~\ref{proposition:2-activating}, we will restrict our attention to $1$-activating adversaries in what follows.
For $1$-activating adversaries, we may refer to the stations participating in an execution by the round numbers in which they got activated.
So when we refer to a \emph{station~$v$}, for some integer $v\ge 0$, then we mean the station that got activated in round~$v$.
If no station got activated in a round~$v$, then a station bearing the number~$v$ does not exist.
To avoid having multiple identities associated with a station, we assume that once a station is activated and later becomes passive, then it never gets activated again; this does not make a difference from the perspective of the adversary, as we assume that there is an unbounded supply of passive stations.

Next we show that, for any injection rate $\rho$, an acknowledgment-based algorithm can be fooled by an adversary of type $(\rho,b)$ for sufficiently large burstiness~$b$.
To this end, it is sufficient to demonstrate the existence of an execution in which two stations simultaneously start working to broadcast a new packet each.
This is because the stations start from the same initial state, they are anonymous, and they  execute the code of the same deterministic algorithm.


\begin{proposition}
\label{pro:acknowledgment-impossibility}

No acknowledgment-based algorithm for channels with collision detection is fair against a $1$-activating adversary of type $(\rho,b)$ such that $2\rho + b\ge 3$.
\end{proposition}

\begin{proof} 
Let us consider an arbitrary acknowledgment-based deterministic distributed algorithm.
We may assume that a passive station immediately attempts to transmit a packet when activated, as otherwise a delay can be offset with a suitably earlier activation.

Let the adversary inject two packets into a passive station in round~$1$, which determines the station number~$1$.
This station transmits its first packet successfully in the second round and next immediately resets its state to initial, because it executes an acknowledgment-based algorithm.
Let the adversary inject one packet into another passive station in round two, which determines the station number~$2$.
The adversary will not inject any other packets.
The two stations, numbered~$1$ and~$2$, start working on new packets from the third round, each being in the initial state.
The stations undergo the same state transitions in the execution that follows,  by induction on round numbers.
In particular, these two stations either transmit simultaneously or pause simultaneously.
This means that each transmission attempt results in a collision, so that the two outstanding packets never get heard on the channel.

For such a scenario to be possible, the adversary of a type $(\rho,b)$ needs to be able to inject three packets in two consecutive rounds.
Let $\tau$ be a time interval of two consecutive rounds.
The adversary can inject $\rho|\tau|+b = 2\rho + b$ packets during~$\tau$, where $|\tau|=2$ is the length of~$\tau$.
For the argument to work, it is sufficient for the inequality $2\rho + b\ge 3$ to hold.
\end{proof} 

Proposition~\ref{pro:acknowledgment-impossibility} can be interpreted as follows: for any acknowledgment-based algorithm and any injection rate, there exists a $1$-activating adversary with this injection rate, and with a suitably large burstiness, such that the algorithm is not fair against this adversary.
Because of this fact, we will consider only activation-based and full-sensing algorithms in what follows.
We investigate the question what is the maximum injection rate for which bounded packet latency can be attained, abstracting from burstiness which will turn out not to be critical. 
An answer may depend on the restrictions determining a class of algorithms, like activation based algorithms, and whether a channel is with collision detection or not.

\Paragraph{Impossibility for injection rate~$1$.}

We observe next that no deterministic distributed algorithm can provide stability against an adversary of type $(\rho,b)$ with the injection rate equal to~$\rho=1$ and with burstiness $b+1\ge 2$.
To see this, consider an arbitrary algorithm executed against such an adversary.
We build an execution by determining prefixes of a sequence of auxiliary executions.
Let an execution $\cE_1$ be obtained by activating a station per each round, by way of injecting one packet into a passive station.
There are the following two cases.
One is when there exists an active station~$v_1$ which was activated with one packet and alone transmits a packet.
The transmission by~$v_1$ is successful in~$\cE_1$ and a packet is heard.
Let us modify execution $\cE_1$ to obtain another execution~$\cE_2$ such that the station~$v_1$ does not get activated at all, which results in this round being silent.
Furthermore, after this silent round in~$\cE_2$, we activate another station  by injecting two packets in it.
The target execution has its prefix determined until and including the simultaneous injection of these two packets in~$\cE_2$.
In the second case, there exist two active stations, say, $v_2$ and~$v_3$, such that they transmit together in~$\cE_1$ in a round of the first transmission in this execution.
This creates a collision, which contributes to a packet delay.
The target execution has its prefix determined until and including this collision in~$\cE_1$.
Which of the two cases holds depends on the algorithm considered.
This construction continues indefinitely producing prefixes of arbitrarily large lengths.
Each time we consider the execution with its prefix determining the target execution, we next examine the suffix after this prefix for one of the two possible cases as above.
The final execution is obtained as the union of these prefixes.
There are infinitely many void rounds in this execution, in which no packet is heard, while the adversary keeps injecting packets, with the rate of one packet per round on the average.
This concludes the argument that the injection rate of one packet per round is too much to facilitate  bounded packet latency by a deterministic distributed algorithm.

This fact  demonstrates a difference between the adversarial model of ad-hoc channels with the model of channels in which  stations know a fixed number of stations attached to the channel along with their names.
In that latter model,  stability can be obtained even for  injection rate~$1$, as it was demonstrated in~\cite{ChlebusKR09}.
Moreover, for that model a bounded packet latency can be attained for any injection rate less than~$1$, see~\cite{AnantharamuCKR-INFOCOM10,AnantharamuCKR-SIROCCO11}.

\Paragraph{Impossibility of injection rate $3/4$.}

We show next that bounded packet latency cannot be obtained by a deterministic algorithm for the model of ad-hoc channels when injection rate equals~$3/4$.
We prepare for the proof of this fact by considering a construction of an execution for any fixed algorithm~$\cA$ and a $1$-activating adversary of a type $(\frac{3}{4},b)$, where $b\ge 2$.

In this execution, the adversary will inject packets into any station only once at the time of this station's initialization.
A station is activated with either one or two packets.
When a station transmits a packet, then we assume that the stations executing the algorithm learn if it was the last packet or there are more pending packet in this station. 
This helps the algorithm and not the adversary.

Should a certain round were not checked for the possibility of a station being activated in this round by providing an opportunity to transmit at least one packet, then there would exist an execution in which a station \emph{is} activated in this round but its packets are never heard on the channel.
This observation is represented formally as follows.

We say that \emph{a round~$v$ is tried in round~$s$} of the  constructed execution if either 
\begin{enumerate}
\item[(1)] a station got activated in round~$v$ and it transmits its packet in round~$s$, or 

\item[(2)] no station got activated in round~$v$, but such a station would have transmitted in round~$s$ if one were activated in  round~$v$.
\end{enumerate}
We use the phrase ``station~$v$ is tried'' interchangeably with ``round~$v$ is tried,'' even if no station is activated in round~$v$ so ``station~$v$'' does not exist.
For a station $v$ to be tried in a round~$s$ we need $s > v$, as it is possible to verify existence of a station activated in a round only after that round.

We say that \emph{a station $v$ gets certified in round $w$} when it is the round in which $v$ is tested and the interaction of $v$ with the channel, or lack thereof, certifies that $v$ does not have  packets pending transmission. 
There are two ways in which such a certification could occur. 
One is when a station identified by the number~$v$ does not exist, because no station got activated in that round, and so station~$v$ cannot transmit.
This is reflected in either a silence or a message transmitted by some other station heard on the channel in this round.
Another way is when $v$ exists and is still active in round $w$, and it transmits its last packet in this round, and  this packet is heard on the channel.
Such a station $v$ immediately become passive and is no longer involved in computations.

When only one station is scheduled to be tried in a round, then we refer to this as a \emph{single trial}, or a \emph{single certification} when the trial results in a certification.
When many stations are scheduled to be tried in one round then this is a \emph{multiple trial}, which also can be a \emph{multiple certification} when there is no collision in this round and a station that transmits, if any, transmits its only packet.

If a station $v$ is certified in a round $s > v$ then the number $s-v$ is called the \emph{delay of certification of~$v$}.
The execution we will construct will have the property that delays of certifications in it grow unbounded while the adversary injects infinitely many packets.
This will allow to conclude that packet latency is unbounded in this execution.
To see this, observe that if algorithm~$\cA$ has packet latency at most~$t$  then the rounds when stations get activated are tried with a delay of at most~$t$. 
This is because if some round~$r$ gets certified later than at round $r+t$ and $r$ got activated that the activating packet waited more than $t$ rounds.
Since the adversary injects infinitely many packets, their stations certifications grow unbounded and the packets wait accordingly.

We organize the construction by considering intervals of rounds that we call \emph{segments}. At any stage of the construction, the already determined segments make a prefix  of the target execution.
The construction proceeds by considering  possible extensions of an already determined prefix of segments by adding a next segment.
This is by way of foreseeing actions of stations, as directed by  algorithm~$\cA$, and involves corresponding actions by the adversary. 
When discussing a next segment, we refer only to stations that still need to complete their certification.


\begin{table}[tp]
\begin{center}
\begin{tabular}{|c| c  c  c  c | c | l |}
\hline
\RB \LB
 \# & \multicolumn{4}{|c|}{ segment}   & delay & remarks \\
\hline
\hline

\RB \LB
 &  &  &  &  & & 
Single certifications of  stations inherited\\
\LB
$C1$ &\underline{$1$} & \ldots & &   & $0$ & 
from previous segments.\\
\hline

\RB \LB
 &  &  &  & &&   
At most $3$ single certifications with no transmissions, \\
\LB
$C2$ &$1$ & $2$ & \ldots &   & $0$ & 
followed by trials of multiple or inherited rounds.\\
\hline

\RB \LB
 & $k$  & & & &  &   
This segment consists of one round.\\
\RB \LB
 & $\vdots$  & & & &  &   
All stations, except for possibly one, are committed.\\
\RB \LB
 & $2$  & & & &  &   
If $k>1$ then the delay is negative,\\
\RB \LB
$C3$ & \underline{$1$}   & & &  & $-k+1$ & 
which means it is a speedup. \\
\hline

\end{tabular}
\parbox{\pagewidth}{\FF\caption{\label{table-a} 
Three special cases of segments. 
}}
\end{center}
\end{table}

Some of such stations may have attempted to transmit during the prefix, but their trial was not completed as a certification because of collisions; such stations are called \emph{inherited} in future segments.
A station that is not inherited is called \emph{new}.
The arguments we give never rely on any stations to be inherited, so activating always costs at least one packet, but if a station is inherited indeed then this may help the adversary in that packets need not be used to activate in order to create a collision.
A  round is \emph{committed} when the adversary determines that no station gets activated in it.

We assume that algorithm~$\cA$ has additional properties that facilitate the exposition of arguments and can be made without loss of generality; we present them next.
Algorithm $\cA$ could waste rounds, like do not attempt trials  in some rounds  or repeat transmissions that already resulted in a collision; we assume that algorithm $\cA$ never behaves this way.
Algorithm~$\cA$ might schedule some rounds to be tried by having stations transmit one station per round, which is neutral with respect to packet delay when the station was not activated at all or when it holds only one pending packet.
Such certifications are neutral with respect to packet delay, as we certify one station in one round and may move to the next station in the next round.
We assume $\cA$ never performs such certifications, unless this is essential in the argument used.
When a new segment is considered then the stations that are tried in it are referred to by the names $1,2,3,\ldots$ which means that these are the stations not certified yet listed in the relative order of their activation times.
This clearly helps algorithm~$\cA$ to decrease worst-case packet latency.
This also facilitates presenting the adversary's strategy, as the order of activation reflects time's flow.
Observe that when a station is tried in a last round of a segment as the smallest station not yet certified and this station transmits successfully in this round but still has a pending packet then this station is the smallest station not yet certified in the next round.
In such a situation, this station is tried in this next round, which is the first round of the next segment, possibly along with other stations.

The adversary's strategy is discussed by cases, which are summarized and illustrated in Tables~\ref{table-a} through~\ref{table-f}.
These cases do not include trials involving more than three stations in one round, which are discussed later and reduced to the tabulated cases.
In these tables, natural numbers are names of stations in the order of joining trials in the segment.
A circled number indicates that the station identified by the number transmits and a collision occurs with another transmission.
An underlined number means  that this station transmits and its last packet is heard.
A doubly underlined number means that the station transmits and its packet is heard but this station holds one more packet.
A number with no marking indicates that this station is scheduled to transmit but has not been activated.
A number with a tilde above means that it is one station of a pair, each of them with a tilde symbol, such that precisely one of these stations holds a packet.
In such a situation of a pair of stations marked by the tilde, the decision which station in the pair holds a packet and which not is made later, depending on the algorithm's actions.
When a station is said to become inherited then this refers to segments following the considered one.


\begin{table}[tp]
\begin{center}
\begin{tabular}{|c| c  c  c  c | c | l |}
\hline
\RB \LB
 \# & \multicolumn{4}{|c|}{ segment} & delay & remarks \\
\hline
\hline

\RB \LB
$C4$ &1 & $2$ & \underline{\underline{$3$}} & \ldots & $1$ & 
A prefix defining this case. 
Stations $1$ and $2$ certified.\\
\hline

\RB \LB
$C5$ &1 & $2$ & \underline{\underline{$3$}} & \underline{$3$}  & $1$ & 
Station $3$ gets certified just after $C4$.\\
\hline

\RB \LB
 &  &  &   &  \circled{$5$} & &  
A possible next round after segment $C5$. 
This is to be \\
\LB
$C6$ & $2$& \underline{\underline{$3$}} &  \underline{$3$} & \circled{$4$} & $2$ & 
 continued as in cases that start with two collisions.\\
\hline

\RB \LB
 &  &  &  &  \circled{$4$}  & &  
Another possible next round immediately after $C4$.\\
\LB
$C7$ &1 & $2$ & \underline{\underline{$3$}} & \circled{$3$}  & $2$ & 
Continued as in cases that start with two collisions.\\
\hline

\RB \LB
 &  &  & \circled{$4$} &  \circled{$5$}  & &  
A possible next round after $C7$.\\
\LB
$C8$ &$2$ & \underline{\underline{$3$}} & \circled{$3$} & \circled{$3$}  & $3$ & 
To be continued as the case of Table~\ref{table-c}. \\
\hline

\RB \LB
 &  &  & \circled{$4$} &  $6$  & &  
A possible next round after $C7$, with the delay including  \\
\LB
$C9$ &$2$ & \underline{\underline{$3$}} & \circled{$3$} & $5$  & $1$ & 
the rounds of $C7$.
To be continued as the case of Table~\ref{table-d}. \\
\hline

\RB \LB
  &  &  &  &  $\tilde{5}$  & &  
A possible continuation of $C4$.\\
\LB
&  &  &  &  $\tilde{4}$ & &  
One among stations $4$ and $5$  holds a packet. \\
\LB
$C10$ &1 & $2$ & \underline{\underline{$3$}} & \circled{$3$}  & $2$ & 
To be continued as the case of Table~\ref{table-e}.
 \\
\hline
\end{tabular}
\parbox{\pagewidth}{\FF\caption{\label{table-b} 
The case when four single trials would occur  if none of the respective stations $1, 2, 3, 4$ were activated.  
To prevent such certifications with no delay, represented as $1\ 2\ 3 \ 4$, station~$3$ is activated with two packets.
Exactly one station with the tilde symbol contains a packet, which one it is to be is decided later.
}}
\end{center}
\end{table}

The execution can make certain actions of the adversary in a given segment either not possible or clearly advantageous, with the goal to increase packet latency.
For instance, when a single station is to be tried in a round then it is advantageous for the adversary not to activate this station, if such a decision is still possible, which is the case when the station is not inherited and not activated as justified by previous events in this segment.
This and similar cases are summarized in Table~\ref{table-a}.
In this table, row~$C1$ represents singe certifications of inherited stations, and the adversary cannot make any decision at this point as the station has been activated already. 
Row $C2$ represents single certifications when there are at most three of them in succession, should none of them were activated.
The adversary does not activate any among these stations.
Row $C3$ depicts one round in which a group of stations get certified simultaneously, each already committed not to be activated in previous segments.
This row is to represent two variants, determined by whether a packet is heard on the channel or not in this round.
As depicted in Table~\ref{table-a}, some station also transmits a packet, which happens to be station~$1$.

A more involved case occurs when at least four consecutive single trials of new uncommitted stations would occur, if none of them were activated; this case is summarized in Table~\ref{table-b}.
This is the only situation in which the adversary activates a station with two packets.
In this case, the first two rounds are silent, which allows the adversary to inject more packets later, which happens in the third round.
Station $3$ transmits in the third round but it has another packet, which defines the prefix defining this case represented by row~$C4$ in Table~\ref{table-b}.
There are four essential continuations possible, depending on what happens in the fourth round of the segment, while we omit continuations with single trials.
One is represented by row $C5$, in which station $3$ transmits again, to become certified thereby.
Another continuation is given in row~$C6$, in which two stations, including station~$3$, are tried again.
Yet another continuation is depicted in row~$C7$, in which two new stations are tried together, both different from~$3$.
This in turn may be continued either as in $C8$, which allows to reduce the problem to Table~\ref{table-c},  or as in $C9$, which can be reduced to the case of Table~\ref{table-d}.
The final continuation is represented by row~$C10$, in which three stations are tried together, which can be reduced to the case of Table~\ref{table-e}.


\begin{table}[tp]
\begin{center}
\begin{tabular}{|c| c  c  c   | c | l |}
\hline
\RB \LB
 \# & \multicolumn{3}{|c|}{ segment}  & delay & remarks \\
\hline
\hline

\RB \LB
 &  \circled{$2$} &  \circled{$3$} & &  &   
A prefix defining this case.\\
\LB
$C11$ &\circled{$1$} &  \circled{$1$} &  \ldots & $2$ & 
Station $4$ is committed now. \\
\hline

\RB \LB
 &  \circled{$2$} &  \circled{$3$} & $4$ &  &   
Stations $1$ and $4$ get certified.\\
\LB
$C12$ &\circled{$1$} &  \circled{$1$} &  \underline{$1$}  & $1$ & 
Stations $2$ and $3$ inherited by future segments.\\
\hline

\RB \LB
 &  \circled{$2$} &  \circled{$3$} & $5$ &   &  
The adversary still can afford to activate\\
\LB
$C13$ &\circled{$1$} &  \circled{$1$} & $4$   & $1$ & 
stations $6$, $7$, and $8$. \\
\hline

\RB \LB
 &   &   & \circled{$5$} &  &   
Station $4$ is committed but not certified yet.\\
\LB
 &  \circled{$2$} &  \circled{$3$} & $4$ &  &   
\\
\LB
$C14$ &\circled{$1$} &  \circled{$1$} & \circled{$1$}   & $3$ & 
The delay is greater than $1$.\\
\hline

\RB \LB
 &   &   & \circled{$6$} &  &   
Station $4$ is committed but not certified yet.\\
\LB
 &  \circled{$2$} &  \circled{$3$} & \circled{$5$} &  &   
\\
\LB
$C15$ &\circled{$1$} &  \circled{$1$} & $4$   & $3$ & 
The delay is greater than $1$.\\
\hline
\end{tabular}
\parbox{\pagewidth}{\FF\caption{\label{table-c} 
The case of segments starting with two trials of pairs of stations, with a repetition of one station among the pairs.  
This makes a prefix of this case specified as $C11$, in which station $1$ represents a repeated station.
}}
\end{center}
\end{table}

Next we consider cases when a segment starts with a multiple trial.
The adversary always activates the two stations of smallest numbers, represented by stations $1$ and $2$, which allows to ignore other stations tried together, if there are any.
All such cases are summarized in Tables~\ref{table-c} through~\ref{table-f}.
They share the same first round, so are determined by what happens in the second round.

The first case among them is summarized in Table~\ref{table-c}.
This case occurs when there is again a multiple trial in the second round and it is such that one of the stations of the first round also participates, which is represented by station~$1$.
Then the adversary also activates the station of the smallest number that is tried in the second round, which is represented by station~$3$.
This specifies the prefix of the first two rounds of this case, as represented by row $C11$ in Table~\ref{table-c}.
At this point, station~$4$ is omitted in activations, and either it is certified, as in segments $C12$ and $C13$ when at most two stations are tried together, or it is not, as in rows $C13$ and $C14$  when three stations are tried together, which allows the adversary to create a collision by activating stations with numbers larger than~$4$.


\begin{table}[tp]
\begin{center}
\begin{tabular}{|c| c  c  c  c | c | l |}
\hline
\RB \LB
 \# & \multicolumn{4}{|c|}{ segment} & delay & remarks \\
\hline
\hline

\RB \LB
 &  \circled{$2$} & $4$ &  &  &&   
A prefix defining this case. \\
\LB
$C16$ &\circled{$1$} & $3$ &\ldots & \ldots & & 
\\
\hline

\RB \LB
 &  \circled{$2$} & $4$ & \circled{$6$} & \circled{$8$} &&   
The maximum number of stations that the adversary \\
\LB
$C17$ &\circled{$1$} & $3$ & \circled{$5$} & \circled{$7$}  &  $2$& 
can activate, as accounted by two stations not activated.\\
\hline

\RB \LB
 &  \circled{$2$} & $4$  & \circled{$6$} & \circled{$7$} &&   
Stations $1,2,5,6,7$ inherited  by future segments.\\
\LB
$C18$ &\circled{$1$} & $3$ & \circled{$5$} & \circled{$1$} & $2$ & 
 \\
\hline

\RB \LB
 &  \circled{$2$} & $4$ & \circled{$5$}  & \circled{$7$} & &  
A case dual to $C18$.\\
\LB
$C19$ &\circled{$1$} & $3$ & \circled{$1$} & \circled{$6$} & $2$ & 
\\
\hline

\RB \LB
 &  \circled{$2$} & $4$ & \circled{$5$} & \circled{$6$} & &   
Repeating a station in trials helps the adversary,\\
\LB
$C20$ &\circled{$1$} & $3$ & \circled{$1$} & \circled{$1$} & $2$ & 
 as compared to the immediately preceding segments.\\
\hline

\end{tabular}
\parbox{\pagewidth}{\FF\caption{\label{table-d} 
The case of a segment that begins with two trials of pairs of stations, with no repetitions among them.
This determines a prefix specified as $C16$. 
The first round  produces a collision.
The second round certifies two rounds with no station activated in them.
}}
\end{center}
\end{table}

The second case starting with a multiple trial if given in Table~\ref{table-d}.
It occurs when there is again a multiple trial of two stations in the second round but no station of the first round participates in it.
These new stations are represented by~$3$ and~$4$.
The adversary does not activate these stations allowing them to be certified in the second round.
This allows the adversary to activate stations $5,6,7,8$ accounting them to not activating $3$ and $4$ while maintaining injection rate~$3/4$.
Possible continuations with multiple trials are represented by rows $C17$ through~$C20$.


\begin{table}[tp]
\begin{center}
\begin{tabular}{|c| c  c  c  c | c |  l |}
\hline
\RB \LB
 \# & \multicolumn{4}{|c|}{ segment} & delay & remarks \\
\hline
\hline

\RB \LB
 &   & $\tilde{4}$ & & & &   
A prefix of this case.
The delay after the first\\
\LB
 &  \circled{$2$} & $\tilde{3}$ &  &  & &   
two rounds.
A dual case is obtained \\
\LB
$C21$ &\circled{$1$} & \circled{$1$} & $\ldots$ & $\ldots$ &  $2$& 
 by replacing $1$ with $2$ in the second round.\\
\hline

\RB \LB
 &   & \circled{$4$} & & & &   
Station $3$ certified before $4$ in a single trial \\
\LB
 &  \circled{$2$} & $3$ &  &  & &   
that follows the first two rounds.\\
\LB 
$C22$ &\circled{$1$} & \circled{$1$} & $3$ & $\ldots$ &  $2$& 
 \\
\hline

\RB \LB
 &   & $4$ & & & &   
A segment dual to $C22$. \\
\LB
 &  \circled{$2$} & \circled{$3$} &  &  & &   
Station $4$ certified before $3$ in a single trial\\
\LB
$C23$ &\circled{$1$} & \circled{$1$} & $4$ & $\ldots$ &  $2$& 
 that follows the first two rounds. \\
\hline

\RB \LB
 &   & \circled{$4$} & & & &   
Station $4$ tried before $3$ along a station\\
\LB
 &  \circled{$2$} & $3$ & \circled{$5$} &  & &   
 of a number greater than $4$, namely $5$. \\
\LB
$C24$ &\circled{$1$} & \circled{$1$} & \circled{$4$} & $\ldots$ &  $3$& 
Station $3$ is committed but not certified yet.\\
\hline

\RB \LB
 &   & $4$ & & & &   
Station $5$ not activated because $4$ is committed.\\
\LB
 &  \circled{$2$} & \circled{$3$} & \circled{$3$} & $5$ & &   
Station $3$ tried before $4$ along an already activated\\
\LB
$C25$ &\circled{$1$} & \circled{$1$} & \circled{$1$} & $4$ &  $2$& 
 station, namely $1$ in this segment. \\
\hline

\RB \LB
 &   & \circled{$4$} & & & &   
Station $4$ used before $3$ along an already activated\\
\LB
 &  \circled{$2$} & $3$ & \circled{$4$} & \circled{$5$} & &   
station, namely $2$.   
Station $3$ is committed\\
\LB
$C26$ &\circled{$1$} & \circled{$1$} & \circled{$2$} & \circled{$4$} &  $4$& 
and not certified yet. 
A delay is greater than $1$. \\
\hline

\RB \LB
 &   & \circled{$4$} & & & &   
A segment dual to $C25$.\\
\LB
 &  \circled{$2$} & $3$ & \circled{$4$} & $5$ & &   
Station $4$ used before $3$ along an already activated\\
\LB
$C27$ &\circled{$1$} & \circled{$1$} & \circled{$2$} & $3$ &  $2$& 
station, namely $2$ in this segment.\\
\hline

\RB \LB
 &   & $4$ & & & &   
 Activation of station $3$ rather than $4$ is arbitrary\\
\LB
 &  \circled{$2$} & \circled{$3$} & $4$ &  & &   
 in this segment, as station $4$ could be activated\\
\LB
$C28$ &\circled{$1$} & \circled{$1$} & \underline{$3$} & $\ldots$ &  $1$& 
 instead. \\
\hline

\end{tabular}
\parbox{\pagewidth}{\FF\caption{\label{table-e} 
The case of a segment that begins with two trials in the first rounds followed by a trial of three  rounds.
There is a repetition of tried stations in the prefix defining this case, unlike in Table~\ref{table-f}. 
Each of the first two rounds produces a collision.
Exactly one station with the tilde symbol contains a packet; it is to be decided later which one it is.
}}
\end{center}
\end{table}

The third case of the first round in a segment being a multiple trial is summarized in Table~\ref{table-e}.
It is defined by a triple trial in the second round when one station of the first trial participates, which is represented by station~$1$, and two new stations are represented by $3$ and~$4$.
The adversary enforces a collision in the second round by activating either $3$ or $4$ but not both; which one it is to be decided later.
This may be determined by a single certification, as exemplified by rows $C22$ and $C23$.
Other sub-cases are listed in rows $C24$ through $C28$, with explanations in the columns of remarks.


\begin{table}[tp]
\begin{center}
\begin{tabular}{|c| c  c  c  c | c |  l |}
\hline
\RB \LB
 \# & \multicolumn{4}{|c|}{ segment} &   delay & remarks \\
\hline
\hline

\RB \LB
 &   & $\circled{5}$ & & & &  
A prefix defining this case.\\
\LB
 &  \circled{$2$} & $\tilde{4}$ & & & &   
Exactly one  among stations $3$ and $4$ \\
\LB
$C29$ &\circled{$1$} & $\tilde{3}$ & \dots & \ldots & $2$& 
is activated with a packet.\\
\hline

\RB \LB
 &   & \circled{$5$} & & & & 
A continuation of $C29$. 
Station $4$  becomes\\
\LB
 &  \circled{$2$} & $4$ & & \circled{$6$} & &   
 committed in the third round because \\
\LB
$C30$ &\circled{$1$} & \circled{$3$} & $4$ &  \circled{$5$}& $3$& 
 it gets tried as single before $3$ does. \\
\hline

\RB \LB
 &   & \circled{$5$} & & & &  
A continuation of $C29$.\\
\LB
 &  \circled{$2$} & \circled{$4$} & & $7$& &   
A segment dual to $C30$, \\
\LB
$C31$ &\circled{$1$} & $3$ & $3$ & $6$ &  $1$& 
with $3$ and $4$ exchanging their roles.\\
\hline

\RB \LB
 &   & \circled{$5$} & & & &  
Station $3$ tried before $4$ along an activated\\
\LB
 &  \circled{$2$} & $4$ & \circled{$3$} & $6$ & &   
station, namely $1$ in this segment.
Station $4$
\\
\LB
$C32$ &\circled{$1$} & \circled{$3$} & \circled{$1$} & \underline{$5$}  & $2$& 
 is committed. 
The delay is greater than $1$.\\
\hline

\RB \LB
 &   & \circled{$5$} & & & &  
A segment dual to $C32$, with $3$ and $4$\\
\LB
 &  \circled{$2$} & \circled{$4$} & \circled{$4$} & $6$& &   
 exchanging roles, and a committed\\
\LB
$C33$ &\circled{$1$} & $3$ & \circled{$1$} & $3$  & $2$& 
station $3$ replacing $5$ in $C32$\\
\hline

\RB \LB
 &   & \circled{$5$} & & & &  
The activation of $3$ rather than $4$ is arbitrary\\
\LB
 &  \circled{$2$} & $4$ & $4$ & \circled{$6$} & &   
 in this segment.\\
\LB
$C34$ &\circled{$1$} & \circled{$3$} & \underline{$3$} &  \circled{$5$}  & $2$& 
Continued as cases starting with two collisions.\\
\hline

\RB \LB
 &   & $\circled{5}$ & & & &  
Station $4$ gets committed in the third round.\\
\LB
 &  \circled{$2$} & $4$ & \circled{$5$} & $6$ & &   
Station $4$ is not certified yet.\\
\LB
$C35$ &\circled{$1$} & \circled{$3$} & \circled{$3$} & \underline{$3$}  & $2$& 
The delay is greater than $1$. \\
\hline

\RB \LB
 &   & $\circled{5}$ & & & &  
This segment is dual to $C35$, with \\
\LB
 &  \circled{$2$} & \circled{$4$} & \circled{$5$} & $6$ & &   
a committed $3$  replacing an activated $3$ in $C35$ \\
\LB
$C36$ &\circled{$1$} & $3$ & \circled{$4$} & $3$  & $2$& 
 in the fourth round. \\
\hline

\end{tabular}
\parbox{\pagewidth}{\FF\caption{\label{table-f} 
The case of a segment with two stations tried in a first round and three stations in the second one.
There is no repetition of tried stations in the prefix defining this case, unlike of Table~\ref{table-e}. 
Each of the first two rounds produces a collision.
Exactly on station with the tilde symbol contains a packet; it is to be decided later which one it is.
}}
\end{center}
\end{table}

The final fourth case with a multiple trial in the first round is specified in~Table~\ref{table-f}.
It is defined by a triple trial in the second round which does not involve a repetition of a station from the first round.
The new stations are represented by $3,4,5$.
The adversary enforces a collision by activating $5$ and exactly one of $3$ and $4$; which one it is depends on what the algorithm makes stations do next.
Again, this may be determined by single certifications in the third round, which is presented in rows~$C30$ and $C31$, the other sub-cases are listed in rows $C32$ through~$C36$.

The construction of an execution discussed above does not cover all possible cases, as it omits simultaneous trials of more than three stations.
Its purpose is to be the main component of the proof of Theorem~\ref{thm:impossibility-for-rate-more-than-3/4} which we give next.


\begin{theorem}
\label{thm:impossibility-for-rate-more-than-3/4}

No deterministic distributed algorithm for channels with collision detection can provide bounded packet latency against a $1$-activating  adversary of injection rate~$\frac{3}{4}$ and with burstiness at least~$2$.
\end{theorem}

\begin{proof} 
Let us consider a specific algorithm~$\cA$.
We argue that the adversary of rate $3/4$ and burstiness $2$ can enforce an execution of~$\cA$ in which delays of certifications grow unbounded and infinitely many packets get injected, while eventually we certify any round.
This is accomplished by constructing an execution applying the adversary's strategy as summarized in Tables~\ref{table-a} through~\ref{table-f}.
We can verify by inspection that the adversary needs burstiness at most~$1$ in Tables~\ref{table-c} through~\ref{table-f} and burstiness~$2$ in Table~\ref{table-b}.

Table~\ref{table-a} is the only one in which delay could be at most~$0$.
Row~$C1$ can occur at most finitely many times, following a given round, determined by the number of inherited stations.
Row~$C2$ consists of at most three rounds.
Row~$C3$ results in a delay, because the cases stipulated in the other tables in which a station gets committed and not certified come with a delay greater than~$1$.
These are the following rows: $C14$ and $C15$ in Table~\ref{table-c}, $C26$ in Table~\ref{table-e}, and their variants.
Therefore even subtracting $1$ from each such a tabulated delay, when $C3$ occurs, keeps  original delays positive, while one station transmitting a packet in~$C3$ provides one extra certification, which is accounted for by the certification's round.

The tabulated cases cover segments with at most three trials in a round.
Segments with more than three activations per round can be replaced with a conceptual execution  such that each multiple trial of at least four stations per round is replaced by a series of trials of pairs or triples per round by partitioning a big set of stations tried in one round into such small $2$-subsets and $3$-subsets.
The adversary applies the same strategy.
We can verify by inspection that each row in Tables~\ref{table-b} through~\ref{table-f} has the property that two consecutive multiple trials result in at least one collision among them.
Therefore, we obtain that the original execution, with trials of  four or more stations in one round, produces a collision during each round of trial of more than three stations per round.

We conclude that indeed the adversary can enforce an execution in which delays of certifications grow unbounded.
\end{proof}

\section{A non-adaptive activation based broadcast}

\label{sec:non-adaptive-activation-based}

We propose a non-adaptive activation-based algorithm, which is called \textsc{Counting-Backoff}.
It is designed for channels with collision detection.
Algorithm \textsc{Counting-Backoff} is based on the idea that active stations maintain a global virtual stack, that is, a last-in-first-out queue.
Each station on the stack remembers its position, which is maintained as a counter with the operations of incrementing and decrementing by~$1$.
A passive or newly activated station has the counter equal to~$0$; such a station will join the stack only when it needs to perform more than one transmissions.
The station at the top of the stack has the counter equal to~$1$.
If a collision of two concurrent transmissions by two stations occurs then the station among these two that was activated earlier gives up temporarily, which means vacating the position at the top of the stack, while the station activated later persists in transmissions, which is implemented as  taking the top position on the stack.


\begin{figure}[t]
\rule{\textwidth}{0.75pt}

\F 
\textbf{Algorithm} \textsc{Counting-Backoff} 

\rule{\textwidth}{0.75pt}
\begin{center}
\begin{minipage}{\pagewidth}
\begin{description}
\Item[\rm /$\ast$] in the round of activation : $\texttt{backoff\_counter} = 0$ $\ast$/ 

\Item[\tt if] \texttt{backoff\_counter} $\leq \, 1$  \texttt{then} transmit

\Item[\tt feedback] $\leftarrow$ feedback from the channel 

\Item[\tt if] \texttt{feedback} = collision  \texttt{then}
		$\texttt{backoff\_counter} \leftarrow \texttt{backoff\_counter} + 1$
		
\Item[\tt else if] \texttt{feedback} = silence  \texttt{then} $\texttt{backoff\_counter} \leftarrow 		\texttt{backoff\_counter} - 1$

\Item[\tt else if] \texttt{feedback} = own message \texttt{then}
\begin{description}

\Item[\tt if] still active  \texttt{then} $\texttt{backoff\_counter} \leftarrow 1$  \texttt{else} $\texttt{backoff\_counter} \leftarrow 0$

\end{description}
\end{description}
\end{minipage}
\FFF

\rule{\textwidth}{0.75pt}

\parbox{\captionwidth}{\caption{\label{alg:stack-backoff}
The code for one round for a station that is active in the beginning of this round. 
A feedback from the channel can be in the form of either a message heard or  silence or  collision.
The command ``transmit'' means transmitting a pending packet, unless there are no such packets, then the station pauses.
}}
\end{center}
\end{figure}

The pseudocode of algorithm \textsc{Counting-Backoff} is presented in Figure~\ref{alg:stack-backoff}. 
Every station has a private integer-valued variable \texttt{backoff\-\_counter}, which is set to~$0$  when the station is passive. 
The private instantiations of the variable \texttt{backoff\_counter} are manipulated by active stations according to the following rules.
An active station transmits a packet  in a round when it \texttt{\texttt{backoff\_counter}} is at most~$1$.
When a collision occurs, then each active station increments its \texttt{backoff\_counter} by~$1$.
When a silent round occurs, then each active station decrements its \texttt{backoff\_counter} by~$1$.
When a message is heard then the counters \texttt{backoff\_counter} are not modified, with the possible exception of a newly activated station.

A station that gets activated has its \texttt{backoff\_counter} equal to~$0$, so the station transmits in the round just after the activation.
Such a station increments its  \texttt{backoff\_counter} in the next round, unless its only packet got heard, in which case the station becomes passive without ever modifying its \texttt{backoff\_counter} and joining the stack.
A station that transmits and its packet is heard continues by withholding the channel and transmitting continuously in the following rounds, unless it does not have any other pending packets or a collision occurs that disrupts transmissions.
Private instances of variable \texttt{backoff\_counter} of acting stations are manipulated such that that are all different and thereby serve as dynamic transient names for stations that are otherwise  nameless.
This makes the number of stations that transmit together to be at most~$2$.

An example of how the algorithm works is presented in Table~\ref{tab:first-example-counting-backoff}.
We continue to use the convention to refer to a station activated in a round $t$ as the station~$t$.
The specific execution presented in Table~\ref{tab:first-example-counting-backoff} is such that station~$1$ is activated with two packets, then station $2$ is activated with one packet, which is followed by station $3$ activated with one packet, and finally station~$5$ is activated with two packets; no other stations get activated and no other packets are injected into active stations.
Details of this execution are as follows.

A station transmits in the next round after becoming activated, following the structure of a round/event as described in Section~\ref{sec:preliminaries}.
Station~$1$ transmits its first packet in round~$2$ and the message is heard on the channel.
Now this station increments $\texttt{backoff\_counter}\leftarrow 1$, to occupy the stack as the only station, and transmits again in round~$3$.
This transmission results in a collision, because station~$2$ got activated in the meantime and transmits its only packet.
The collision results in both stations~$1$ and~$2$ incrementing their variables \texttt{backoff\_counter} to~$2$ and~$1$, respectively.
This means that station~$1$ will pause in the next round while station~$2$ will transmit again.
A collision occurs in round~$4$ because station~$3$ got activated in the meantime and transmits its only packet in round~$4$.
This collision results in all the active stations incrementing their \texttt{backoff\_counter} variables to occupy three different positions on the stack in round~$5$.
The repeated transmission of station~$3$ is heard in round~$5$, because station~$4$ does not exist as no station got activated in round~$4$.
In round~$6$, a packet transmitted by a newly activated station~$5$ is heard, which results in stations~$1$ and~$2$ keeping their positions on the stack.
Observe that in this round the top position on the stack is not occupied while there are two stations on the stack.
Station~$5$ occupies the top of the stack in round~$7$ when it transmits its second packet.
Round~$8$ is the fourth consecutive one with stations~$1$ and $2$ staying on positions $3$ and $2$ of the stack, respectively.
The silent round~$8$ results in stations $1$ and $2$ shifting their positions by~$1$ towards the top of the stack, which results in station~$2$ transmitting its only packet in round~$9$.
The silent round~$10$ results in station~$1$ decrementing $\texttt{backoff\_counter}\leftarrow 1$ and transmitting its second packet in round~$11$.


\begin{table}[t]
\begin{center}
\begin{tabular}{| c ||c |c |c |c |c | c | c | c | c | c | c | c | }
\hline
\RB \LB
\# & $1$ & $2$ & $3$ &$4$&$5$& $6$ & $7$& $8$& $9$ & $10$ & $11$ & $12$ \\
\hline\hline
\RB \LB
$1$& $2\boldsymbol{a}$  & $h/0$ & $c/1$ & $s/2$& $s/3$&$s/3$& $s/3$& $s/3$& $s/2$ &$s/2$&$h/1$&$\boldsymbol{p}/0$\\
\hline
\RB \LB
$2$ & & $1\boldsymbol{a}$ & $c/0$ &$c/1$& $s/2$& $s/2$& $s/2$&$s/2$&$h/1$ &$\boldsymbol{p}/0$&&\\
\hline
\RB \LB
$3$ & & & $1\boldsymbol{a}$& $c/0$&$h/1$ &$\boldsymbol{p}/0$ &&& &&&\\
\hline
\RB \LB
$4$ & & && $\boldsymbol{\square}$ & &&&& &&&\\
\hline
\RB \LB
$5$ & & &&&$2\boldsymbol{a}$ &$h/0$&$h/1$&$\boldsymbol{p}/0$& &&&\\
\hline
\end{tabular}
\parbox{\pagewidth}{
~
\caption{\label{tab:first-example-counting-backoff}
An example of an execution of algorithm \textsc{Counting-Backoff}.
Row $i$ represents activity of station~$i$.
Column $k$ represents actions of all stations in round~$k$.
Symbol $2\boldsymbol{a}$ represents activation with two packets and $1\boldsymbol{a}$ represents activation with one packet.
Symbol~$\boldsymbol{\square}$ at row~$i$  and column~$i$ means that station $i$ is not activated.
Symbol $h/x$ means that the station begins the round with $\texttt{backoff\_counter} = x$, it  transmits a packet and the message is heard.
Symbol $c/x$ means that the station begins the round with $\texttt{backoff\_counter} = x$, it  transmits a packet and this results in a collision.
Symbol $s/x$ means that the station begins the round with $\texttt{backoff\_counter} = x$ and stays silent.
Symbol $\boldsymbol{p}/0$ means that the station becomes passive after resetting $\texttt{backoff\_counter} \leftarrow 0$.
}}
\end{center}
\end{table}

Next we discuss the correctness and performance of algorithm  \textsc{Counting-Backoff}.


\begin{lemma}
\label{lem:stack}

When an active station executing \textsc{Counting-Backoff} has its \emph{\texttt{backoff\_counter}} positive at the beginning of a round, then this value may be interpreted as  this station's position on a global stack of stations, with the active station whose $\emph{\texttt{backoff\_counter}} =1$  placed at the top.
\end{lemma}

\begin{proof} 
We argue that the following invariant is maintained in any round of an execution of algorithm \textsc{Counting-Backoff}.
\begin{quote}
\textsf{The invariant:}
If there are some $k>0$ stations active in the beginning of a round then the following three properties hold at the end of this round:
\begin{enumerate}
\item[(i)] 
each station that remains active has a different number from the interval $[1,k]$ stored in its \texttt{backoff\_counter};

\item[(ii)]
the increasing values of  \texttt{backoff\_counter} correspond to the inverse order of  activation of the stations;

\item[(iii)] 
$1$ is the only number in the interval $[1,k]$ that might not be assigned as a name in a round.
\end{enumerate}
\end{quote}
This invariant is shown by induction  on  round numbers.
The base of induction holds because there are no active stations in the beginning.
Consider an arbitrary round $t+1>1$ and assume that the invariant holds prior to this round.
In round~$t+1$,  either a packet is heard, or the round is silent, or there is a collision. 
Next we consider each of these three cases separately. 

When a packet is heard in round~$t+1$, then two sub-cases arise, depending on whether some station was activated in round~$t$ or not.

The first sub-case occurs when station~$t$ got activated in round~$t$.
This station transmits a packet in round~$t+1$ because its \texttt{backoff\_counter} equals~$0$.
Then, because a packet is heard at round~$t$ and by the inductive assumption, either the stack is empty or the position~$1$ on top of the stack is not occupied, as otherwise the station at the top  would have \texttt{backoff\_counter} equal to~$1$ and would have transmitted and created a collision.
If station~$t$ is still active after the transmission, then $t$ sets its \texttt{backoff\_counter} to~$1$ and becomes the first station on the stack, otherwise either the stack remains empty or its top position remains unoccupied.

The second sub-case occurs when no station got activated in round~$t$.
Then, because a packet is heard and by the inductive assumption, the station at the top of the stack transmitted in round $t+1$.
If this station remains active after the transmission, then nothing changes in the arrangement of stations on the stack.
If the transmitting station becomes passive, then it resets \texttt{backoff\_counter} back to zero, which is interpreted as this station leaving the stack, while no other station claims the top position,  so $1$ is not used as a value of \texttt{backoff\_counter} at this moment.

The next case occurs when round $t+1$ is silent.
This means that no station has \texttt{backoff\_counter} less than or equal to~$1$.
So no station got activated in round~$t$ and either the stack is empty or the smallest value of  \texttt{backoff\_counter} of a station on the stack is~$2$. 
In the latter case,  each station on the stack decrements its \texttt{backoff\_counter} by~$1$, making its value the true position on the stack.

The final case is of a collision in round~$t+1$.
This means that some station has its \texttt{backoff\_counter} equal to~$1$, and so it is at the stack's top, while another station has its variable \texttt{backoff\_counter} equal to~$0$, which means that this is the station newly activated in round~$t$.
Now each active station increments its \texttt{backoff\_counter} by~$1$, which results in inserting the station~$t$ on top of the stack.
\end{proof} 

By Lemma~\ref{lem:stack}, the packet held by the station at the bottom of the stack is delayed longer than any packet in a station that shares the stack with the bottom one.
Therefore, a bound on time it takes for the stack to stay nonempty, from adding the first station to an empty stack until this station gets removed from the stack, is a bound on packet latency.
To guarantee that the stack will eventually become empty, after a station joins it, the injection rate needs to be small enough.

We may observe that to maintain the size of the stack, it is sufficient to inject a packet only in every third round.
This is because it results in a collision in every third round, a message heard in every third round, and silence in every third round.
In such an execution, the adversary may activate station~$1$ with two packets, then station $2$ with one packet, and continue activating stations $3i-1$, for $i=1,2,3\ldots$ with one packet each.
The first round is silent, there is a message heard in the second round, and next a message is heard in each round  $3i+1$, for $i=1,2,3\ldots$.
The beginning of such an execution for $b=3$ is depicted in Table~\ref{tab:example-counting-backoff-lack-fairness}.
The effect is that the stack never gets empty and the packet hold by the station at its bottom is never heard. 
It follows that algorithm \textsc{Counting-Backoff} is not fair when injection rate is $\frac{1}{3}$ and $b\ge 3$.


\begin{table}[t]
\begin{center}
\begin{tabular}{| c ||c |c |c |c |c | c | c | c | c | c | c |  }
\hline
\RB \LB
\# & $1$ & $2$ & $3$ &$4$&$5$& $6$ & $7$& $8$& $9$ & $10$ & $11$  \\
\hline\hline
\RB \LB
$1$& $2\boldsymbol{a}$  & $h/0$ & $c/1$ & $s/2$& $s/2$&$c/1$& $s/2$& $s/2$& $c/1$ &$s/2$&$s/2$\\
\hline
\RB \LB
$2$ & & $1\boldsymbol{a}$ & $c/0$ &$h/1$& $\boldsymbol{p}/0$& & &&&&\\
\hline
\RB \LB
$3$ & & & $\boldsymbol{\square}$  & & & &&& &&\\
\hline
\RB \LB
$4$ & & && $\boldsymbol{\square}$ & &&&& &&\\
\hline
\RB \LB
$5$ & & &&&$1\boldsymbol{a}$ &$c/0$&$h/1$&$\boldsymbol{p}/0$& &&\\
\hline
\RB \LB
$6$ & & &&  & &$\boldsymbol{\square}$&&& &&\\
\hline
\RB \LB
$7$ & & &&  & &&$\boldsymbol{\square}$&& &&\\
\hline
\RB \LB
$8$ & & &&  & &&& $1\boldsymbol{a}$ & $c/0$&$h/1$&$\boldsymbol{p}/0$\\
\hline
\end{tabular}
\parbox{\pagewidth}{
~
\caption{\label{tab:example-counting-backoff-lack-fairness}
A beginning of an execution of algorithm \textsc{Counting-Backoff}.
The notational conventions are the same as in Table~\ref{tab:first-example-counting-backoff}.
The injection rate is~$\frac{1}{3}$.
When this pattern of injections is continued indefinitely, then station~$1$ stays on the stack forever.
}}
\end{center}
\end{table}

Next we argue that algorithm \textsc{Counting-Backoff} has bounded packet latency when executed against an  adversary of  injection rate less than~$\frac{1}{3}$.
To this end, consider packets in the station at the bottom of the stack, as specified in Lemma~\ref{lem:stack}.
Let us consider a time interval~$\tau$ such that a station is added to an empty stack in the first round of~$\tau$.
We may partition the interval~$\tau$ into contiguous segments of rounds such that a collision marks  the end of a segment.
It is convenient to assume, assume without loss of generality, that a successful transmission results in removing a station from the stack, with the only exception being at the very beginning of~$\tau$, in which a station is activated with at least two packets.
Next we review what happens during segments depending on their length.

A segment of length~$1$  means a collision just after a collision in the previous round.
This results in a newly-activated station added to the stack.

A segment of length~$2$ results in the stack maintaining its size.
This occurs except before the fist collision, when the stack is empty, because we hear a message after a collision, which is transmitted by the station at the top of the stack, and then a new collision indicates that a station has just been activated, which takes the vacant top position on the stack.
An example of an execution in which the stack's size fluctuates between $1$ and $2$ when segments are of size~$2$, except for the very beginning, is given in Table~\ref{tab:example-counting-backoff-stack-does-not-grow}.

A segment of length~$3$ may result in the stack maintaining its size, as it may consist of a successful transmission, a silent round, and a collision.
This is visualized in Table~\ref{tab:example-counting-backoff-lack-fairness}.

A segment of length at least~$4$ results in the stack shrinking in size, because it begins with a successful transmission and  includes at least another successful transmission.


\begin{table}[t]
\begin{center}
\begin{tabular}{| c ||c |c |c |c |c | c | c | c | c | c | }
\hline
\RB \LB
\# & $1$ & $2$ & $3$ &$4$&$5$& $6$ & $7$& $8$& $9$ & $10$ \\
\hline\hline
\RB \LB
$1$& $2\boldsymbol{a}$  & $h/0$ & $c/1$ & $s/2$& $s/2$&$s/2$& $c/1$& $s/2$& $s/2$ &$s/2$\\
\hline
\RB \LB
$2$ & & $1\boldsymbol{a}$ & $c/0$ &$h/1$& $\boldsymbol{p}/0$& & &&&\\
\hline
\RB \LB
$3$ & & & $\boldsymbol{\square}$  & & & &&& &\\
\hline
\RB \LB
$4$ & & &&  $1\boldsymbol{a}$ &$h/0$ &$\boldsymbol{p}/0$&&& &\\
\hline
\RB \LB
$5$ & & && & $\boldsymbol{\square}$ &&&& &\\
\hline
\RB \LB
$6$ & & &&  & &  $1\boldsymbol{a}$ &$c/0$&$h/1$& $\boldsymbol{p}/0$&\\
\hline
\RB \LB
$7$ & & &&  & &&$\boldsymbol{\square}$&& &\\
\hline
\RB \LB
$8$ & & &&  & &&& $1\boldsymbol{a}$ & $h/0$&$\boldsymbol{p}/0$\\
\hline
\end{tabular}
\parbox{\pagewidth}{
~
\caption{\label{tab:example-counting-backoff-stack-does-not-grow}
An example of an execution of algorithm \textsc{Counting-Backoff} possible when injection rate is~$\frac{1}{2}$.
The notational conventions are the same as in Table~\ref{tab:first-example-counting-backoff}.
When this pattern of injections is continued indefinitely then the stack contains either one or two stations.
The execution cycles through a repetition of a pattern occurring in any four consecutive rounds, starting from round~$2$.
}}
\end{center}
\end{table}

When the stack shrinks in a segment, then it is most conducive to increasing packet latency when the segment is of length~$4$ rather than larger than~$4$.
We may assume conservatively that after the stack becomes non-empty then the adversary tries to maintain this property by enforcing as many segments of length~$3$ as possible but no shorter segments.
Because the injection rate is less than~$\frac{1}{3}$, the number of segments of length~$4$ grows as the length of interval~$\tau$ increases so that eventually the stack becomes empty.
The sufficient length of an interval~$\tau$ for the stack guaranteed to become empty during~$\tau$  depends only on the adversary's type, so indeed packet latency of algorithm \textsc{Counting-Backoff} is bounded.

Our approach to provide explicit bounds on packet latency is to discover an adversarial strategy that maximizes packet latency, depending on the adversary's type.
An adversary's strategy is called \emph{stack persistent}, during a time interval,  when it has the following properties.
\begin{quote}
\begin{enumerate}
\item
When a station is activated while no other station is active, then two packets are injected into this station, and in the next round another station is activated.
\item
If a station is activated while other stations are still active, then this is by way of injecting just one packet into the station.
\item
If a station is activated by injecting a packet while other stations are still active, then the first attempt to transmit this packet results in a collision. 
\item
When a message is heard in a round while some other stations are active, then the next round is silent.
\end{enumerate}
\end{quote}
We say that an action of the adversary is advantageous, as compared to other ones, when it results in increased packet latency. 

Next we argue that a stack persistent strategy can maximize packet latency of algorithm \textsc{Counting-Back\-off} when it is executed against an adversary of a type $(\rho,b)$ such that $\rho < \frac{1}{3}$ and $b\ge 3$. 
To push a station on an empty stack, this station needs to be initialized with multiple packets, and another station needs to be initialized sufficiently early in the immediately following rounds.
This creates a collision and the former station gets pushed on the stack.
It is advantageous to push the first station on the stack by activating it with two packets only, and then, while the stack is nonempty, to activate the next stations by injecting one packet per station.
This is because then each station on the stack holds one packet, so it takes at least one extra round between two successful packet transmissions.
When the stack is nonempty then it is advantageous to have any injection result in a collision.
Similarly, when a message is heard in a round then it is advantageous to have this round followed by a silent round.
When such a strategy is applied, we can associate three rounds with each packet injected during the interval, namely, the round of a collision when it is first transmitted, the round when it is heard, and the following silent round.
The only exceptions are the first round of the interval, when the stack is created, and the round of the last transmission, which makes the stack empty again.
This gives a tight upper bound on the duration of the interval, with respect to the number of packets injected during the interval.


\begin{theorem}
\label{thm:stack-with-collision-detection}

When algorithm \textsc{Counting-Backoff} is executed on a channel with collision detection against an adversary of type $(\rho,b)$, where $\rho<\frac{1}{3}$ and $b\ge 3$, then packet latency is at most $\frac{3b-3}{1-3\rho}$ and the number of packets queued in any round is at most $\frac{3b-3}{2}$.
\end{theorem}

\begin{proof} 
Let $\tau$ be a time interval such that a station is activated in the first round of~$\tau$ and this station's last packet is transmitted in the last round of $\tau$, where the length $|\tau|$ of $\tau$ is greatest with this property achievable by the adversary.
We may assume that $\tau$ begins in the first round of the execution, as then the adversary is not additionally constrained by previous injections.
Let the adversary apply a strategy that is most conducive to packet delay.
It could be the following specific strategy.
First a stack of two stations is created, and then a new station is activated  three rounds after the last activation, as long as this is possible.
This means that the beginning of the execution looks like that with a pattern depicted in Table~\ref{tab:example-counting-backoff-lack-fairness}

It follows that if $k$ packets are injected in~$\tau$, then $\tau$ consists of $k$ rounds in which  messages are heard, $k-1$ silent rounds, and $k-2$ rounds with collisions, for a total of $3k-3$ rounds.
Let $t=|\tau|$ be the length of~$\tau$.
By the specification of the adversary, at most $\rho t + b$ packets are injected in~$\tau$.
The worst case occurs when all these packets are heard in $\tau$, so that $k=\rho t + b$.
We obtain a system of two equations, 
\begin{eqnarray*}
t &=& 3k-3\\
k&=&\rho t + b\ , 
\end{eqnarray*}
in the variables $t$ and~$k$.
The equations yield a solution $t=\frac{3b-3}{1-3\rho}$.

Next we discuss a strategy to maximize queue size.
To make the stack grow continuously, the adversary needs to inject packets in a sequence of consecutive rounds, because injecting it even in every other round may result in the stack not growing, as visualized in Table~\ref{tab:example-counting-backoff-stack-does-not-grow}.
This results in collisions, starting from the second such an injection, each collision incrementing the number of stations on the stack.
Interspersing such a sequence of injections with breaks does not help in the stack growing, as just maintaining the size of the stack requires injecting with frequency one in three while the injection rate is less than~$\frac{1}{3}$.

Therefore, a strategy to maximize queue size is to make an interval in which injections occur as big as possible.
When this strategy is applied, for the maximum possible number~$L$ of rounds, then $L$ satisfies the equality 
\begin{equation}
\label{eqn:queue-size-counting-backoff}
L=b-1 + \rho L
\ , 
\end{equation}
so that $L=\frac{b-1}{1-\rho}$.
The number $ L$ is an upper bound on  the size of the stack, because one packet is heard in the second round.
We can bound $L$ by the inequality $L< \frac{3}{2} (b-1)$, because $1/(1-\rho)<\frac{3}{2}$ when  $\rho<\frac{1}{3}$ holds.
\end{proof}

The  bound on packet latency of algorithm \textsc{Counting-Backoff} given in Theorem~\ref{thm:stack-with-collision-detection} is tight, as it is obtained by estimating the delay during adversarial strategy that is most conducive to packet delay.
It follows that  packet latency grows unbounded when the injection rate $\rho$ approaches  $\frac{1}{3}$.
On the other hand, the bound on queue size given in Theorem~\ref{thm:stack-with-collision-detection} depends only on the burstiness of the adversary.
The upper bound $\frac{3b-3}{2}$, on the number of queued packets at any round, holds also when injection rate equals~$\frac{1}{3}$, because equation~\eqref{eqn:queue-size-counting-backoff} applies with $\rho=\frac{1}{3}$ and then $\frac{3b-3}{2}$ is the solution.
We conclude that algorithm \textsc{Counting-Backoff} is stable when injection rate equals~$\frac{1}{3}$, with $\frac{3b-3}{2}$ as an upper bound on queue size,  while it is not fair.

\section{A non-adaptive full sensing broadcast}

\label{sec:full-sensing}

Stations executing a full-sensing algorithm can listen to the channel at all times, thereby maintaining a sense of time through references to events in past rounds.
Such a sense of time gives a potential for an active station to interpret the round of activation as an explicit identity.

A broadcast algorithm may process consecutive past rounds to give the stations activated in them an opportunity to transmit.
This, just by itself, may result in unbounded packet latency, if we spend at least one round to examine any past round for a possible activation in it.
This is because a recurring occurrence of active stations with multiple packets each would accrue unbounded delays.
To prevent this, one may consider groups of multiple rounds and have stations activated in these rounds transmit simultaneously. 
If at most one station got activated in a group then we save at least one round of examination, which compensates for the delays due to occasionally some stations holding more than one packet.
If a channel is with collision detection, then this property helps to implement such an approach.
We assume in this section that channels are with collision detection.

We refer to active stations by the respective rounds of their activation, as before.
A round gets \emph{verified} when either all the packets of the station activated in this round have been heard or when it becomes certain that no station got activated in this round.

Next we present a non-adaptive full-sensing algorithm which we call \textsc{Quadruple-Round}.
The rounds of an execution are partitioned into disjoint groups of four consecutive rounds, each called a \emph{segment}.
The first and second rounds of a segment make its \emph{left pair}, while the third and fourth rounds make the \emph{right pair}.
The rounds of execution spent on processing the rounds in a segment make a  \emph{phase} corresponding to this segment.
The purpose of a phase is to verify the stations in the corresponding segment.
A pseudocode for a phase is given in Figure~\ref{alg:quadruple-round}.


\begin{figure}[t]
\rule{\textwidth}{0.75pt}

\F 
\textbf{Algorithm} \textsc{Quadruple-Round} 

\rule{\textwidth}{0.75pt}
\begin{center}
\begin{minipage}{\pagewidth}

\begin{description}

\Item[\tt repeat] ~
\B
\begin{description}

\Item[\rm transmit] /$\ast$ first round $\ast$/

\Item[\tt feedback] $\leftarrow$ feedback from the channel 

\Item[\tt if] \texttt{feedback} = silence \texttt{then terminate} /$\ast$ end of phase $\ast$/

\Item[\tt else if] \texttt{feedback} = message heard  \texttt{then exit} /$\ast$ end of iteration $\ast$/

\Item[\tt else] /$\ast$ a collision in the first round $\ast$/

\begin{description}

\Item[\tt if] a station in the left pair \texttt{then} transmit  /$\ast$ second round $\ast$/

\Item[\tt feedback] $\leftarrow$ feedback from the channel 

\Item[\tt if] \texttt{feedback} = silence \texttt{then} /$\ast$ no  station in the left pair is active $\ast$/

\begin{description}

\Item[\tt if] the first station in the right pair \texttt{then} transmit in the third round

\Item[\tt else if] the second station in the right pair \texttt{then} transmit in the fourth round

\end{description}

\Item[\tt else if] \texttt{feedback} = message heard \texttt{then exit} /$\ast$ end of iteration $\ast$/

\Item[\tt else] /$\ast$ a collision so both stations in the left pair are active $\ast$/

\begin{description}

\Item[\tt if] the first station in the left pair \texttt{then} transmit in the third round

\Item[\tt else if] the second station in the left pair \texttt{then} transmit in the fourth round

\end{description}

\end{description}

\end{description}
\BB
\Item[\tt until] the phase is terminated

\end{description}

\end{minipage}
\FFF

\rule{\textwidth}{0.75pt}

\parbox{\captionwidth}{\caption{\label{alg:quadruple-round}
A code for a phase determined by a segment.
A phase is structured as a loop which repeats ``iterations.''
There are four stations in a segment, which are partitioned into left and right pairs.
The command ``transmit'' means transmitting a packet from the private queue, unless the queue is empty.
The command ``exit'' ends an iteration but not the phase.
The round numbers refer to the rounds of iteration.
There may be at most four rounds per one iteration of the repeat loop.}}
\end{center}
\end{figure}

To implement a sense of time, it is not necessary to maintain a counter of the verified rounds, which  would grow unbounded.
Instead, one may count the number of rounds since the latest round examined for a station activated in it.
With such an implementation, when packet latency is bounded in an execution of a broadcast algorithm, then the values of the private variables, which are used to implement the sense of time in this way, are bounded as well.
This mechanism of implementing time is used in algorithm \textsc{Quadruple-Round} but it is omitted from  the pseudocode in Figure~\ref{alg:quadruple-round}, which concentrates on the schedule of transmissions.
Another aspect of implementing the sense of time is that a phase to verify the rounds of a segment starts only after at least four rounds have passed since the first round of a segment, otherwise the phase is delayed for as long as needed to have this condition satisfied.

A phase of algorithm \textsc{Quadruple-Round} is organized as a loop, which repeats actions to which we collectively refer as an iteration of the loop.
It takes at most four rounds to perform an iteration, by a direct inspection of the pseudocode in Figure~\ref{alg:quadruple-round}.
An iteration is executed as follows.
All the stations activated in the rounds of the phase's segment, if there are any, transmit together in the first round of an iteration.
A station, that is scheduled to transmit, transmits a packet from its private queue, unless the queue  is empty.
This results in either a silence or a message heard or a collision. 
Next we discuss the corresponding three cases.

When the first round of an iteration is silent, then this ends the iteration and also the phase.
This is because such a silence confirms that there are no outstanding packets in the active stations  in the segment.

When a message is heard in the first round of an iteration, then this ends the iteration but not the phase.
The reason of continuing the phase is that the station which transmitted the packet heard on the channel may have more packets.

If a collision occurs in the first round of an iteration, then the stations of the left pair transmit together in the second round.
This leads to three sub-cases presented next.

The first sub-case is of silence in the second round of the iteration, which means that no station in the left pair is active.
As the first round in this phase  produced a collision, this means that each station in the right pair holds a pending packet.
In this sub-case, the third and fourth rounds of the iteration are spend by the third and fourth stations of the segment transmitting one packet each in order, which concludes the iteration but not the phase.

The second sub-case is of a message heard in the second round, which concludes the iteration but not the phase.
This means that exactly one station in the left pair holds packets.
The phase is continued to verify if the station which transmitted the packet heard on the channel has more packets.

The third sub-case occurs when there is a collision in the second round of the iteration, which means that each station in the left pair of the segment holds an outstanding packet.
In this case, the third and fourth rounds are spend by the first and second stations of the segment transmitting one packet each in order, which concludes the iteration but not the phase.

The segments of an execution of the algorithm are partitioned into disjoint pairs of two consecutive segments, any such a group called a \emph{double segment}.
The two phases corresponding to a double segment make a \emph{double phase}.


\newlength{\examplewidth}
\setlength{\examplewidth}{\textwidth}
\addtolength{\examplewidth}{-20em}

\begin{figure}[t]
\begin{center}
\begin{minipage}{\examplewidth}
\begin{align}
\label{eqn:first-example}
 - - - X &\ \ \ \ - - - -  \\
 \label{eqn:second-example}
- - - - &\ \ \ \  X X - - \\
\label{eqn:third-example}
X - X X &\ \ \ \ - - - -  \\
\label{eqn:fourth-example}
X X - -  &\ \ \ \ - - X X \\
\label{eqn:fifth-example}
 - - X X &\ \ \ \ X - X X   \\
 \label{eqn:sixth-example}
-  - X X &\ \ \ \ X X X X  \\
\label{eqn:seventh-example}
X -  X^3 X^3 &\ \ \ \ - - - - \\
\label{eqn:eighth-example}
- - - -  &\ \ \ \  X^4 X^4 - -
\end{align}
\end{minipage}

\FFF

\parbox{\captionwidth}{\caption{\label{fig:quadruple-examples}
Examples of double segments determining double phases.
An example is represented by a horizontal string of eight symbols, each symbol corresponding to a round.
A dash stands for a round with no station activated in it.
Symbol~$X$ stands for a round in which a station was activated and it holds one packet.
Expression~$X^\ell$ represent a round in which a station was activated and holds $\ell$ packets.
}}
\end{center}
\end{figure}

We give examples giving the maximum number of rounds that a double phase takes, depending on the number of packets held by active stations in the corresponding double segment.
When there are no active stations in a double segment, then the first phase takes~$1$ round and the second one as well, by the code in Figure~\ref{alg:quadruple-round}, for a total of two rounds.
When the number of packets in active stations are between~$1$ and~$8$, then the respective configurations are depicted in Figure~\ref{fig:quadruple-examples}.
Next we discuss how the execution of algorithm \textsc{Quadruple-Round} proceed in the examples in Figure~\ref{fig:quadruple-examples}, by inspecting and following the pseudocode of Figure~\ref{alg:quadruple-round}.

The first phase of example~\eqref{eqn:first-example} begins by having a packet heard, which is followed by silence, to give the station whose message was heard a possibility to transmit other packets in case there were any. 
The second phase produces only a silent round.
We obtain that double phase~\eqref{eqn:first-example} consists of three rounds in total.

The first phase in example~\eqref{eqn:second-example}  is just one silent round.
The second phase consists of collision, then another collision, then a packet heard, then a packet again, and finally silence.
We obtain that double phase~\eqref{eqn:second-example} consists of six rounds.

The first phase in example~\eqref{eqn:third-example} produces collision, then a packet is heard, then there is another collision, silence, a packet, a packet, and a silent round, in the  given order.
The second phase produces a silent round.
This means that this double phases takes eight rounds, so that it verifies eight rounds in eight rounds.

The first phase in example~\eqref{eqn:fourth-example} consists of  collision, then silence, a packet heard, another packet, and the closing silence.
The second phase in this example begins with a collision, which is followed by silence, then two packets heard in two consecutive rounds, and the closing silence.
We obtain that this double phase  takes ten rounds.

The double phase of example~\eqref{eqn:fifth-example} has its first phase as the second phase in example~\eqref{eqn:fourth-example}, which produces five rounds, and the second phase as the first phase in example~\eqref{eqn:third-example}, which produces seven rounds, so double phase~\eqref{eqn:fifth-example} takes twelve rounds in total. 

The first phase of example~\eqref{eqn:sixth-example} is the same as in example~\eqref{eqn:fifth-example}, so it takes five rounds.
The second phase begins with collision, which is followed with another collision and two packets heard, then again two collisions and two packets heard, followed by the closing silence, for the total of nine round.
This means that  double phase~\eqref{eqn:sixth-example} takes fourteen rounds.

The first  phase in example~\eqref{eqn:seventh-example} starts with a collision followed by a packet heard, which completes the first iteration,  then collision, silence, and two packets heard in the second iteration, the third and fourth iteration iterations are similar, which is followed by the closing silence, for a total of fifteen rounds.
The second phase consists of one silence.
The whole double phase~\eqref{eqn:seventh-example} thus takes sixteen  rounds.

The first phase in example~\eqref{eqn:eighth-example} takes one silent round.
The second phase starts with four identical iterations, each consisting of collision,  silence, and two packets heard, followed by the closing silent round, for the total of seventeen rounds.
Thus the whole double phase~\eqref{eqn:eighth-example} takes eighteen rounds.

Now we are ready to argue that if the adversary is of injection rate greater than $\frac{3}{8}$ then the packet latency of algorithm \textsc{Quadruple-Round}  is unbounded.
Such an adversary can inject at least three packets in each double segment.
If three packets are injected exactly as in example~\eqref{eqn:third-example} in Figure~\ref{fig:quadruple-examples}, then the corresponding double phase takes eight rounds.
Each such a double phase results in packet latency staying intact, up to a $\cO(1)$ variation, because eight rounds are spent to verify eight rounds.
Additionally, the adversary can inject four packets into a double segment infinitely many times in an execution.
If such injections are performed similarly as in example~\eqref{eqn:fourth-example} in Figure~\ref{fig:quadruple-examples}, then the corresponding double phases takes ten rounds.
Each such a double phase contributes to increasing packet latency by two rounds, because ten rounds are spent to verify eight rounds.
Unlike algorithm \textsc{Counting-Backoff}, algorithm \textsc{Quadruple-Round} is fair for any injection rate.
This is because each phase eventually ends with all the rounds in the corresponding segment verified.

\begin{lemma}
\label{lem:linear-time}

Let active stations in a double segment hold $k$ packets in total, for $k\ge 0$. 
If $k=0$ then the corresponding double phase takes~$2$ rounds, and when $k=1$ then the double phase takes~$3$ rounds.
If $k > 1$, then the corresponding double phase takes at most $2k+2$ rounds, while it may take exactly $2k+2$ rounds, for some configurations of packets held by active stations in the double segment.
\end{lemma}

\begin{proof}
A double phase with no packets to transmit consists of two iterations of a silent round each.
A double phase with just one packet injected into an active station, in the corresponding double segment, is similar to the first double phase in example~\eqref{eqn:first-example} in Figure~\ref{fig:quadruple-examples}, which consists of three rounds.

When an iteration takes four rounds then this results in two packets heard, in the third and fourth rounds of the iteration, by the pseudocode in Figure~\ref{alg:quadruple-round}.
It is impossible for an iteration to take three rounds, by inspection of the pseudocode in Figure~\ref{alg:quadruple-round}.
When an iteration lasts two rounds, then the first round produces a collision and a packet is heard in the second round.
There two possibilities for a one-round iteration: one consists of just one packet heard and the other 
 of a silent round that closes a phase. 
The only possible way to have a phase terminated is to have a one-round iteration of a silent round.
This means that the number of rounds in iterations in which packets are heard is as most $2k$, and there are also two silent iterations closing the phases.

To have $2k+2$ rounds  spent to hear $k$ packets, for $k>1$, we  specify a pattern of injections generalizing some of the examples in Figure~\ref{fig:quadruple-examples}.
Let $i$ be such that either $k=2i$ or $k=2i+1$, for $i>0$.
In the former case, which is like example~\ref{eqn:eighth-example}, we activate the fifth and sixth stations in a double segment with $i$ packets each.
In the latter case, which like example~\ref{eqn:seventh-example}, we activate the first station with one packet, and next the third and fourth stations with  $i$ packets each.
A direct inspection, similar to one used in discussion examples in Figure~\ref{fig:quadruple-examples}, shows that $2k+2$ rounds are spent to hear these $k$ packets.
\end{proof}

Now we are ready to argue that algorithm \textsc{Quadruple-Round} has bounded packet latency and queues when executed against an adversary of injection rate $\frac{3}{8}$.


\begin{theorem}
\label{thm:quadruple-round}

When algorithm \textsc{Quadruple-Round} is executed on a channel with collision detection against an adversary of type $(\frac{3}{8},b)$, then packet latency is at most $2b +4$ and there are at most $b + \cO(1)$ packets queued in any round.
\end{theorem}

\begin{proof} 
Let us consider an arbitrary interval of $\ell$ double phases during which $k$ packets need to be heard on the channel.
By Lemma~\ref{lem:linear-time}, this takes at most $2(k+\ell)$ rounds.
This bound is independent of the distribution of the $k$ packets among the  $\ell$ considered double segments.
Because of this property, we may conceptually distribute the $k$ packets in a balanced manner among the $\ell$ segments.
This allows to abstract from the specific values of $k$ and $\ell$ and instead resort to the specification of the adversary.
Except for the burstiness component of the adversary, represented by the number~$b$ in its type, the adversary may inject three packets per a double segment of eight rounds on the average.
It follows, by Lemma~\ref{lem:linear-time}, that each of the considered double phases takes at most eight rounds, which is the length of a double segment.
This means that packets are heard as they are injected, and any local increase of packet latency in an execution is due to the burstiness component of the adversary. 
When the adversary injects up to $b$ packets in a round, then each of them  increases packet latency by two rounds, by Lemma~\ref{lem:linear-time}, for a total of $2b$ rounds.
Additionally, the algorithm waits for the first $4$~rounds to start processing the first segment.
By the same argument, $b+\cO(1)$ packets remain queued in any round.
\end{proof}

\section{An adaptive  activation based broadcast}

\label{sec:adaptive-activation-based}

Adaptive algorithms may use control bits in messages.
We present an adaptive activation-based algorithm which we call \textsc{Queue-Backoff}.
It is based on the idea that active stations maintain a global virtual first-in-first-out queue.
This approach is implemented such that if a collision occurs, caused by two concurrent transmissions, then the station activated earlier persists in transmitting while the station activated later gives up temporarily.
This is a dual alternative to the rule used in algorithm \textsc{Counting-Backoff}.


\begin{figure}[t]
\rule{\textwidth}{0.75pt}

\F 
\textbf{Algorithm} \textsc{Queue-Backoff} 

\rule{\textwidth}{0.75pt}
\begin{center}
\begin{minipage}{\pagewidth}
\begin{description}
\Item[\rm /$\ast$] in a round of activation : 
\hfill 
$\texttt{queue\_size} =  \texttt{queue\_position} =  \texttt{collision\_count} = 0 \ \ast$/
\Item[\tt if] $0\le \texttt{queue\_position} \le 1$   \texttt{then} transmit 
\Item[\tt feedback] $\gets$ feedback from the channel 
\Item[\tt if] \texttt{feedback} = message \texttt{and} $\texttt{queue\_position} = 0$ \texttt{and} still active \texttt{then} 
\hfill 
/$\ast$  own message  $\ast$/
\begin{description}
\Item[\texttt{queue\_position}] $\gets 1$ \ \ ; \ \ $\texttt{queue\_size} \gets 1$ \ \ 
\hfill
/$\ast$ station joins empty queue $\ast$/
\end{description}
\Item[\tt if] \texttt{feedback} = message with $Q > 0$ 
		\texttt{and} \texttt{queue\_position} $= -1$ \texttt{then} /$\ast$  foreign message  $\ast$/
\begin{description}
\Item[\texttt{queue\_size}] $\gets Q$\ \ ; \ \ 
\texttt{queue\_position} $\gets Q +1 - \texttt{collision\_count} $ 
\end{description}
\Item[\tt if] \texttt{feedback} = message with an attached ``over''  bit set on  \texttt{then} 

\hfill	
/$\ast$ queue size and positions in queue already known to all stations in the queue $\ast$/
\begin{description}
\Item[\texttt{queue\_size}] $\gets \texttt{queue\_size} - 1$
\ \ ; \ \ 
$\texttt{queue\_position} \gets \texttt{queue\_position} - 1$ 
\end{description}
\Item[\tt if] \texttt{feedback} $\ne$ collision \texttt{then} 
\texttt{collision\_count} $\gets 0$ ;
\Item[\tt if] \texttt{feedback} = collision \texttt{then}
\begin{description}
\Item[\texttt{collision\_count}] $\gets \texttt{collision\_count} + 1$ 
\Item[\tt if] \texttt{queue\_size} $> 0$  \texttt{then}
		 \texttt{queue\_size} $\gets \texttt{queue\_size} + 1$	 
\item[\tt else]  $\texttt{queue\_position} \gets -1$ 
\hfill 
/$\ast$ to mark unknown queue size and position $\ast$/
\end{description}
\end{description}
\end{minipage}
\FFF

\rule{\textwidth}{0.75pt}

\parbox{\captionwidth}{\caption{\label{alg:queue-backoff}
The algorithm code for one round of an active station. 
The command ``transmit'' means transmitting a message with a packet along with \texttt{queue\_size} and ``over'' bit attached.
When the last packet in the private queue is transmitted then the  ``over'' bit is set on otherwise it is set off.
The number~$Q$ denotes the value of the sender's variable \texttt{queue\_size} attached to a message.
}}
\end{center}
\end{figure}

We assume first that a channel is with collision detection.
A pseudocode of algorithm \textsc{Queue-Backoff} is in Figure~\ref{alg:queue-backoff}.
Every station has three private integer-valued variables: \texttt{queue\_size}, \texttt{queue\_position}, and \texttt{collision\_count}, which are all set to~$0$ in a passive station.
The values of these variables are related to a station's knowledge about the global distributed virtual queue of stations and this station's position in the queue.
In particular, \texttt{queue\_position} equals $0$ in a newly activated station that is not in the queue yet, and value $1$ means that the station is at the front position in the queue.
A  transmitted message includes the following three components: 
\begin{enumerate}
\item[(1)] a packet, 
\item[(2)] the value of the sender's private variable \texttt{queue\_size}, and 
\item[(3)] an ``over'' bit.
\end{enumerate}
The ``over'' bit in a message is either set on or off; it is set on when the packet in the message was the last one in the sender's queue.
In a round, an active station  transmits a message when its \texttt{queue\_position} equals either~$0$ or~$1$.

Algorithm \textsc{Queue-Backoff} is driven by feedback from the channel.
A silent round occurs only when it is the round of activation of the only active station.
A newly activated station's actions do not depend on a feedback, which is why the event ``\texttt{feedback} = silence'' is not listed explicitly in Figure~\ref{alg:queue-backoff}.
The private variables are modified according to the following rules.

The variable \texttt{queue\_position} is initialized to~$0$ but when a station is activated then it updates this variable in the next round as follows: if the own transmitted message is heard then the station becomes the only station in the queue so \texttt{queue\_position} is set to~$1$, otherwise, when collision occurs, then \texttt{queue\_position} is set to $-1$ as a marker that the queue size is unknown.
When a message with some value $Q>0$ of \texttt{queue\_size} is heard and an active station has $\texttt{queue\_position}=-1$, then the station sets   $\texttt{queue\_position}\gets Q+1 -\texttt{collision\_count}$.
When a message with an ``over'' bit set on is heard, then each active station decrements its variables \texttt{queue\_position}  by~$1$; this operation may be performed in the same round in which the queue size was learned.

The variable \texttt{queue\_size} is initialized to~$0$ and is set to $1$ when a station transmits a message immediately after activation and this message is heard.
Otherwise, when collision occurs in such a transmission, then the station waits until a message with some value $Q>0$ of \texttt{queue\_size} is heard, which results in the station setting $\texttt{queue\_size}\gets Q$.
A positive value of the variable \texttt{queue\_size} is decremented after a message  is heard with an ``over'' bit set on.
When a collision occurs, then each active station with a positive value of \texttt{queue\_size} increments its \texttt{queue\_size} by~$1$. 

The variable \texttt{collision\_count} stores the length of the latest streak of contiguous collisions which includes the preceding round.
It is initialized to~$0$ and is incremented in each round of a collision.
This variable is reset back to~$0$ in a round in which a message is heard.


\begin{table}[t]
\begin{center}
\begin{tabular}{| c ||c |c |c |c |c | c | c | c | c | c | }
\hline
\RB \LB
\# & $1$ & $2$ & $3$ &$4$&$5$& $6$ & $7$& $8$& $9$ & $10$ \\
\hline\hline
\RB \LB
$1$& $3\boldsymbol{a}$  & $h/\textrm{q} 0 $ & $c/\textrm{q} 1 $ & $c/\textrm{q} 2 $ &$h/\textrm{q} 3 $ & $c/\textrm{q} 3 $ &  $h/\textrm{q} 4 $& $\boldsymbol{p}$& &  \\
&  & $ \textrm{t} 0 / \textrm{z} 0$ & $ \textrm{t} 1 / \textrm{z} 0$ & $ \textrm{t} 1 / \textrm{z} 1$ &$ \textrm{t} 1 / \textrm{z} 2$ & $\textrm{t} 1 / \textrm{z} 0$ & $\textrm{t} 1 / \textrm{z} 1$ & & &  \\
\hline
\RB \LB
$2$ & & $1\boldsymbol{a}$ &  $c/\textrm{q} 0$ & $s/\textrm{q} 0 $ & $s/\textrm{q} 0$ &$s/\textrm{q} 3 $ & $s/\textrm{q} 4$ &$h/\textrm{q} 3$ &$\boldsymbol{p}$&\\
& & &  $ \textrm{t} 0 / \textrm{z} 0$ & $ -\textrm{t} / \textrm{z} 1$ & $ -\textrm{t} / \textrm{z} 2$ &$\textrm{t}2 / \textrm{z} 0$ & $\textrm{t}2 / \textrm{z} 1$ &$\textrm{t}1 / \textrm{z} 0$&&\\
\hline
\RB \LB
$3$ & & & $1\boldsymbol{a}$  & $c/\textrm{q} 0 $ & $s/\textrm{q} 0 $ & $s/\textrm{q} 3$ &$s/\textrm{q} 4$&$s/\textrm{q} 3$&$h/\textrm{q} 2$ &$\boldsymbol{p}$\\
& & & & $ \textrm{t} 0 / \textrm{z} 0$ & $-\textrm{t} / \textrm{z} 1$ & $ \textrm{t}3 / \textrm{z} 0$  &$ \textrm{t}3 / \textrm{z} 1$&$\textrm{t}2 / \textrm{z} 0$&$\textrm{t}1 / \textrm{z} 0$ &\\
\hline\RB \LB
$4$ & & &  & $\boldsymbol{\square}$  & &&&& &\\
\hline
\RB \LB
$5$ & & && & $1\boldsymbol{a}$ &$c/\textrm{q} 0 $ &$s/\textrm{q} 0$ &$s/\textrm{q} 3$& $s/\textrm{q} 2$&$h/\textrm{q} 1$\\
 & & && & &$\textrm{t} 0 / \textrm{z} 0$ &$-\textrm{t}  / \textrm{z}1 $&$\textrm{t} 3  / \textrm{z}0 $& $\textrm{t} 2  / \textrm{z}0 $&$\textrm{t} 1  / \textrm{z}0 $\\
\hline
\end{tabular}
\parbox{\pagewidth}{
~
\caption{\label{tab:example-queue-backoff}
An example of a beginning of execution of algorithm \textsc{Queue-Backoff}.
Row $i$ represents activity of station~$i$.
Column $k$ represents actions of all stations in round~$k$.
Symbol $3\boldsymbol{a}$ means activation with $3$ packets and $1\boldsymbol{a}$ activation with $1$ packet.
Symbol~$\boldsymbol{\square}$ at row~$j$  and column~$j$ means that station $j$ is not activated.
Symbol $h/\textrm{q}i$  means that the station transmits and a message is heard and the station has $\texttt{queue\_size}=i$.
Similarly, $c/\textrm{q} i $ means a transmission with a collision, and $s/\textrm{q} i $ means pausing.
Symbol $ \textrm{t} j / \textrm{z} k$ means that the station has $\texttt{queue\_position} =j$ and $\texttt{collision\_count} = k$, for $j\ge 0$, and $-\textrm{t}/\textrm{z} k$ means that $\texttt{queue\_position} =-1$.
Symbol $\boldsymbol{p}$ means becoming passive.
}}
\end{center}
\end{table}

An example of an execution of algorithm \textsc{Queue-Backoff} is given in Table~\ref{tab:example-queue-backoff}.
In this execution, only stations $1,2,3$ and~$5$ are ever activated.
Station~$1$ is activated with three packets.
One of them is heard in the second round, which results in the station taking the front position in the queue.
Rounds $3$ and $4$ produce collisions, to the effect that only station~$1$ knows the size of the queue while stations $2$ and $3$ enter round $5$ with $\texttt{queue-position}=-1$.
Station $1$ attaches the queue size~$3$ in the message heard in round~$5$, which makes station $2$ and $3$ learn this information and determine their proper positions in the queue, namely, positions $2$ and~$3$.
Round~$6$ produces a collision so station~$5$ learns that the queue is nonempty of unknown size, which is marked by setting $\texttt{queue-position}\gets -1$.
Station~$1$ manages to have its last packet heard in round~$7$, which also allows all the station to learn of the size~$4$ of the queue and their positions in it.
The ``over'' bit in the last message of station~$1$ makes all the stations reduce the queue size and decrement their positions in the same round, so that round $8$ is entered with station~$2$ at the front.
The queue shrinks in each of the next two steps to become empty in round~$10$.

Some of the mechanisms of this algorithm work similarly to those in algorithm \textsc{Counting-Backoff}.
In particular, a station that becomes active transmits in the next round after activation. 
Also, when a station transmits and the transmitted message is heard, then the station withholds the channel by transmitting in the following rounds, subject to packet availability. 
Similarly, a collision in a round means that some new station got activated in the previous round. 
This is because a station that has transmitted multiple times, with no other stations intervening, was at the front of the queue with \texttt{queue\_position} equal to~$1$, while any other option triggering transmitting is to have the variable \texttt{queue\_position}  equal to~$0$.
This is only possible when the value \texttt{queue\_position}~=~$0$ is inherited from the state of being a passive station, because \texttt{queue\_position} is reset from $0$ to $-1$ upon the collision that occurs immediately after the activation.
A difference with algorithm \textsc{Counting-Backoff} is that an active station cannot receive silence as feedback from the channel.
This is because \textsc{Queue-Backoff} is adaptive and the ``over'' bit in messages eliminates silent rounds when there are still active stations in the queue.


\begin{lemma}
\label{lem:queue}

When the value of \emph{\texttt{queue\_position}} of an active station executing \textsc{Queue-Backoff} is positive then this value may be interpreted as the number of this station's position in a global first-in-first-out queue of stations, where the active station with $\emph{\texttt{queue\_position}} = 1$ is at the front of the queue.
\end{lemma}

\begin{proof} 
There are two invariants that hold true in any execution of algorithm \textsc{Queue-Backoff}.
\begin{quote}
\textsf{The first invariant:} stations whose variable \texttt{queue\_size} is positive have the value of this variable equal to  the number of active stations.
\end{quote}
The proof of this invariant is by induction on round numbers.
When the first active station stays active for at least two rounds, then this results in its variable $\texttt{queue\_size}$ becoming equal to~$1$.
 The inductive step is by the rules of manipulation of the instance of this variable, namely, collision increases it by~$1$  and the ``over'' bit set on decreases the value.

We define the \emph{virtual position} of a station in a round next, where we refer to this station's private variables in this round.
It is a positive integer equal to either the value of \texttt{queue\_position}, when it is positive, or the number $k +1 - \texttt{collision\_count}$, where $k$ is the number of active stations in this round.

\begin{quote}
\textsf{The second invariant:} at the end of a round, each active station has a different virtual position and such that the virtual positions of all active stations fill the interval $[1,k]$, where $k$ is the number of all active stations.
 \end{quote}
The proof of this invariant is by induction on round numbers.
We consider cases depending on feedback from the channel.
Silent rounds occur when there are no active stations activated prior to the current round.
Once a feedback from the channel is different from silence, then it is either a message heard on the channel or collision.
A message brings the number of active stations in it.
By the inductive assumption, the update of the variables \texttt{queue\_position} at active stations that have heard only collisions up to now is correct.
When collisions occur, then they make a contiguous sequence, and we may carry out induction on the length of such a sequence.
The base of induction is when the first collision occurs.
The newly activated station has its collision count equal to~$1$ at this point, which cancels the~$1$ in the expression $k +1 - \texttt{collision\_count}$.
So the newly activated station obtains the virtual position~$k$ as the only station.
Each subsequent collision does not affect the virtual position of a station that has  \texttt{collision\_position} equal to~$-1$ because both the number of active stations and  the count of collisions increase by~$1$.

The two invariants imply that once a message is heard and there are $k$ active stations then this message propagates the true number of active stations and each station makes it virtual position to be the value of the variable \texttt{queue\_position}, with these values filling the interval $[1,k]$ and assigned in the order of activation. 
\end{proof} 

We say that a station \emph{gets enqueued} when it makes its \texttt{queue\_position} different from~$0$.
Such a station may not know its position in the queue for some time, but this position is already uniquely determined as virtual position, by the invariants used in the proof of Lemma~\ref{lem:queue}.
A station \emph{is in the queue} starting from getting enqued until the station's last packet is heard.


\begin{lemma}
\label{lem:queue-transmits}

In any round, exactly one station in the queue associated with algorithm \textsc{Queue-Backoff} transmits, unless the queue is empty.
A collision occurs only in a round following activation of a new station and such that the queue is nonempty.
\end{lemma}

\begin{proof}
A station whose \texttt{queue\_position} equals~$1$ transmits in a round.
We observe that the station at the front of the queue knows that its position is~$1$.
This can be shown by induction on round numbers during an evolution of a queue while it is nonempty.
The base of induction holds because when the first station in enqueued then it sets its position in the queue to~$1$, by the pseudocode in Figure~\ref{alg:queue-backoff}.
The first station gets enqueued because it has multiple packets.
Transmissions of the first station continue until a message is heard.
This message disseminates the size of the queue among all the other active stations, so all of them, if any, can compute their respective positions in the queue.
A station at the queue's front leaves only after all its packets have been heard.
At this point the station at position~$2$, if any, determines that its new position is~$1$, by the pseudocode in Figure~\ref{alg:queue-backoff}.
The only possibility for a collision is when one active station has \texttt{queue\_position} equal to~$1$ while another, newly activated, has it still equal to~$0$.
\end{proof}

Unlike algorithm \textsc{Counting-Backoff}, algorithm \textsc{Queue-Backoff} is fair for any injection rate less than~$1$.
This is because when a packet is injected then either it is heard in the next round or the station of injection gets attached to the end of the queue. 
While the station holding a packet is in the queue, there is a finite number of other packets that are scheduled to be heard before.
With injection rate less then~$1$, there are infinitely many rounds in which no packet is injected, and so, by Lemma~\ref{lem:queue-transmits}, sufficiently many opportunities to hear all packets that are closer to the front of the queue.

When the injection rate is greater than $\frac{1}{2}$ than packet latency of algorithm \textsc{Queue-Backoff} is unbounded.
This is because when the adversary initializes new stations as often as possible then the rate of the queue growing at the rear surpasses its rate of it shrinking at the front.
Next, we show that when the injection rate is $\frac{1}{2}$ then packet latency is bounded, and derive tight upper bound on the queue size and packet latency.

We say that a \emph{packet~$p$ meets packet~$q$} when either (1) $p$ and $q$ are held by the same station or (2) $q$ is already in the queue when the station that holds~$p$ is enqueued.
We say that a \emph{packet~$p$ meets station~$v$} when station $v$ gets enqueued while $p$'s station is already enqueued.

\begin{lemma}
\label{lem:meeting-packets}

A packet's delay during an execution of algorithm \textsc{Queue-Backoff} is at most the number of packets and stations that the packet meets, plus~$2$.
\end{lemma}

\begin{proof}
When a packet $p$ does not meet any other packet nor station then it is  heard in the next round after injection, so the packet spends just one round in the system.
Otherwise, let a packet~$p$ be injected into a station that is either already enqueued or will be enqueued in the next round after $p$'s injection.
Packet $p$ meets every packet that is already in the queue when $p$ is injected and also any packet injected into the same station along with~$p$.
We may conservatively assume that $p$ is heard as the last packet transmitted by the station holding~$p$.
Each of these packets~$q$ that $p$ meets delays~$p$ by one round, namely the round when $q$ is heard.
Packet $p$ meets every station that gets enqueued after $p$'s station got enqueued while $p$ is waiting in the queue.
Such an addition to the queue is marked by a collision, occurring just after the added station gets activated, which  contributes one round to $p$'s waiting time.
The two possible extra rounds of $p$ waiting, unaccounted for yet, consist of (1)~collision when $p$ is transmitted in vain just after its injection, and (2)~the round when $p$ is finally heard.
\end{proof}

Next we estimate packet latency when injection rate is~$\frac{1}{2}$.
An adversary's strategy is called \emph{queue persistent}  during a time interval  when it has the following properties.
\begin{quote}
\begin{enumerate}
\item
When a station is activated, while no other station is active, then two packets are injected into this station, and in the next round another station is activated.
\item
If a station is activated, while some other stations are active, then this is by way of injecting one packet into the station.
\item
When a message is heard in a round, and there are still pending packets in stations in the queue,  then collision occurs in the next round.
\end{enumerate}
\end{quote}
An action of the adversary is considered advantageous, as compared to some other ones, when the action increases packet latency. 

Next we argue that a queue-persistent strategy can maximize packet latency of algorithm \textsc{Queue-Backoff} for an adversary of a type $(\frac{1}{2},b)$. 
To start the queue, a passive station is activated with multiple packets and another station in the next round again, so a collision occurs in the following round.
When a station is activated, while some stations are still active, then the first attempt to transmit this packet in the next round results in a collision, by Lemma~\ref{lem:queue-transmits}.
Therefore, when some stations are in the queue, each round is either a collision or some packet is heard, with no intervening silent rounds.
It is advantageous for the adversary to activate the first station with two packets only and the following stations with just single packets, in the time interval when a queue persistent strategy is applied, because then, except for the first packet heard, each injection of a packet contributes a collision.
When each station in the queue stores just one packet, and a station is activated with only one packet while the queue is nonempty, then hearing a message decreases the queue by one station and a collision results in adding one station to the queue. 
It is advantageous for the adversary to  create a collision at least in every other round, which results in the queue either oscillating between two consecutive sizes or growing.


\begin{theorem}
\label{thm:queue-algorithm}

When algorithm \textsc{Queue-Backoff} is executed against an adversary of type $(\frac{1}{2}, b)$, then packet latency is at most~$4b-4$ and there are at most $2b-3$ packets in the queue in any round.
\end{theorem}

\begin{proof} 
We consider a queue-persistent strategy that is most conducive to packet delay.
Let $p$ be a specific test packet.
By Lemma~\ref{lem:meeting-packets}, this packet's latency is maximized by the number of stations and packets that~$p$ meets.
This is obtained in two stages.
The first stage is building a queue of the largest possible size, which stores the packets that the test packet~$p$ is to meet.
This is followed by injecting the test packet~$p$, and next, in the second stage, making~$p$ meet as many stations as possible.

This strategy is implemented as follows by the adversary.
In the first round, the adversary injects two packets into the first station.
In subsequent rounds, the adversary activates a station per round by injecting one packet into it.
This continues for a maximum possible number~$y$ of rounds, where $y$ satisfies the equality $y=b-1 + \frac{1}{2}y$, so that $y=2 (b-1)$.
The number of packets  after these $y$ rounds is $y-1=2b-3$. 
This number $y-1$ is also an upper bound on the number of packets queued at any time.
Next, the adversary injects with maximum possible frequency, which means in every other round.
Consider a test packet $p$ injected into a passive station when there are $y-1=2b-3$ stations in the queue.
This packet $p$ needs two rounds to move one position closer to the front.
Concurrently, a station is activated with one packet in these two rounds.
This means that packet~$p$ meets $2b-3$ stations added to the queue while $p$ is already in the queue.
The total number of packets and stations that packet~$p$ meets is $4b-6$.
By Lemma~\ref{lem:meeting-packets}, this packet's latency is $4b-4$. 
\end{proof}

Algorithm \textsc{Queue-Backoff} was presented as implemented for channels with collision detection.
By Lemma~\ref{lem:queue-transmits}, when the global queue is nonempty, then each round contributes either a collision or a message heard on a channel.
This means that when a channel is without collision detection, then collisions can be detected as void rounds by any involved active station, while passive stations do not participate anyway.
It follows that this algorithm can be executed on channels without collision detection with only minor modifications in code and with the same performance bounds.

\section{Conclusion}

We introduced ad hoc multiple access channels along with an adversarial model of packet injection,  in which deterministic distributed algorithms can handle non-trivial injection rates.
These rates make the increasing sequence of numbers $\frac{1}{3}$, $\frac{3}{8}$ and $\frac{1}{2}$.
To go beyond rate $\frac{1}{3}$, which is attained by an activation-based non-adaptive algorithm, we designed a non-adaptive algorithm that handles injection rate~$\frac{3}{8}$ and an adaptive one  that handles injection rate~$\frac{1}{2}$.
The optimality of these algorithms in their respective classes, in terms of magnitude of injection rates  that can be  handled with bounded packet latency against $1$-activating adversaries, is open.

Our non-adaptive algorithms are designed for channels with collision detections.
It is open if non-adaptive algorithms can handle \emph{any} positive injection rate in a stable manner on channels without collision detection.

We showed that no algorithm can handle an injection rate higher than~$\frac{3}{4}$ in a manner that provides bounded packet latency.
It is an open question if any injection rate in interval $(\frac{1}{2},\frac{3}{4})$ can be handled with bounded packet latency by deterministic distributed algorithms against $1$-activating adversaries.

Deterministic algorithms cannot cope against adversaries that can activate multiple stations in a round in a fair manner, as no deterministic distributed algorithm is fair against a $2$-activating adversary of burstiness at least~$2$.
It is an interesting question to ask if randomization can help to handle such adversaries in a stable manner, possibly for sufficiently small injection rates, with stability understood in a suitably defined statistical sense,  for instance as discussed in~\cite{Chlebus-randomized-radio-chapter-2001}.

\pagebreak


\bibliographystyle{abbrv}

\bibliography{mac-adhoc}

\end{document}